\documentclass[12pt]{article}

\newcommand{\blind}{0}

\usepackage[english]{babel} 
\usepackage{amsmath,amssymb,amsthm,eucal} 
\usepackage{enumerate} 
\usepackage{hyperref} 
\usepackage{natbib} 
\usepackage{multirow,arydshln} 
\usepackage{graphicx,color,epstopdf} 
\usepackage[explicit]{titlesec} 
\usepackage{titling} 

\newtheorem{theorem}{Theorem}
\newtheorem{definition}{Definition}
\newtheorem{lemma}{Lemma}
\renewenvironment{proof}[1][\proofname]{{\noindent \bfseries #1.}}{\qed \vspace{0.5cm}}

\def\kron{\hspace{0.5mm} \otimes \hspace{0.5mm}} 
\def\diff(#1,#2){\displaystyle \frac{\partial #1}{\partial #2}} 
\providecommand{\abs}[1]{\lvert#1\rvert} 
\providecommand{\norm}[1]{\lVert#1\rVert} 

\graphicspath{{graphics/}}

\def\spacingset#1{\renewcommand{\baselinestretch}%
{#1}\small\normalsize} \spacingset{1}

\titleformat{\section}
  {\Large \centering}{\thesection.}{1em}{\MakeUppercase{#1}}
\titleformat{\subsection}
  {\large}{\thesubsection}{1em}{#1}

\begin{document}

\if0\blind
{
  \title{\bf Robust Inference for Seemingly Unrelated Regression Models}
  \author{Kris Peremans and Stefan Van Aelst \hspace{.2cm} \\
    Department of Mathematics, KU Leuven, 3001 Leuven, Belgium \\
  }
  \maketitle
} \fi

\if1\blind
{
  \vspace*{0.35cm}
  \begin{center}
    {\LARGE \bf Robust Inference for Seemingly Unrelated Regression Models\par}
  \end{center}
  \vspace*{3.5cm}
} \fi

\bigskip
\begin{abstract} 
\noindent Seemingly unrelated regression models generalize linear regression models by considering multiple regression equations that are linked by contemporaneously correlated disturbances. Robust inference for seemingly unrelated regression models is considered. MM-estimators are introduced to obtain estimators that have both a high breakdown point and a high normal efficiency. A fast and robust bootstrap procedure is developed to obtain robust inference for these estimators. Confidence intervals for the model parameters as well as hypothesis tests for linear restrictions of the regression coefficients in seemingly unrelated regression models are constructed. Moreover, in order to evaluate the need for a seemingly unrelated regression model, a robust procedure is proposed to test for the presence of correlation among the disturbances. The performance of the fast and robust bootstrap inference is evaluated empirically in simulation studies and illustrated on real data. 
\end{abstract}

\noindent KEYWORDS: 
Diagonality test; Fast and robust bootstrap; MM-estimator; Robust testing
\vfill

\noindent \copyright{} 2018. This manuscript version is made available under the CC-BY-NC-ND 4.0 license \url{http://creativecommons.org/licenses/by-nc-nd/4.0/}

\newpage
\spacingset{1.45} 


\section{Introduction}

\noindent Many scientists have investigated statistical problems involving multiple linear regression equations. Unconsidered factors in these equations can lead to highly correlated disturbances. In such cases, estimating the regression parameters equation-by-equation by, e.g., least squares is not likely to yield efficient estimates. Therefore, seemingly unrelated regression (SUR) models have been developed. SUR models take the underlying covariance structure of the error terms across equations into account. Applications in econometrics and related fields include demand and supply models~\citep{Kotakou2011,Martin2007}, capital asset pricing models~\citep{Hodgson2002,Pastor2002}, chain ladder models~\citep{Hubert2017,Zhang2010}, vector autoregressive models~\citep{Wang2010}, household consumption and expenditure models~\citep{Kuson2012,Lar2011}, environmental sciences~\citep{Olaolu2011,Zaman2011}, natural sciences~\citep{Cadavez2012,Hasenauer1998} and many more.

A SUR model, introduced by~\citet{Zellner1962}, consists of $m>1$ dependent linear regression equations, also called blocks. Denote the $j$th block in matrix form by
\[ y_j = X_j \beta_j + \varepsilon_j, \]
where $y_j = (y_{1j},\ldots,y_{nj})^\top$ contains the $n$ observed values of the response variable and $X_j$ is an $n \times p_j$ matrix containing the values of $p_j$ input variables. Note that the number of predictors does not need to be the same for all blocks. The vector $\beta_j = (\beta_{1j},\ldots,\beta_{p_jj})^\top$ contains the unknown regression coefficients for the $j$th block and $\varepsilon_j = (\varepsilon_{1j},\ldots,\varepsilon_{nj})^\top$ constitutes its error term. The error term $\varepsilon_j$ is assumed to have ${\rm E}[\varepsilon_j] = 0$ and ${\rm Cov}[\varepsilon_j] = \sigma_{jj} I_n$ where $\sigma_{jj}$ is the unknown variance of the errors in the $j$th block, and $I_n$ represents the identity matrix of size $n$. In the SUR model blocks are connected by the assumption of contemporaneous correlation. That is, the $i$th element of the error term of block $j$ may be correlated with the $i$th element of the error term of block $k$. With $i$ and $\ell$ observation numbers and $j$ and $k$ block numbers, the covariance structure of the disturbances can be summarized as
\begin{align*}
{\rm E}[\varepsilon_{ij} \varepsilon_{ik}] &= \sigma_{jk}, \quad i=1,\ldots,n \text{ and } j,k=1,\ldots,m; \\
{\rm E}[\varepsilon_{ij} \varepsilon_{\ell j}] &= 0, \quad i \neq \ell; \\
{\rm E}[\varepsilon_{ij} \varepsilon_{\ell k}] &= 0, \quad j \neq k \text{ and } i \neq \ell.
\end{align*}
Note that each regression equation in a SUR model is a linear regression model in its own right. The different blocks may seem to be unrelated at first sight, but are actually related through their error terms.

The regression equations in a SUR model can be combined into two equivalent single matrix form equations. Let bdiag$()$ denote the operator that constructs a block diagonal matrix from its arguments. Moreover, let $\kron$ denote the Kronecker product and let $\Sigma$ be a symmetric matrix with elements $\sigma_{jk}$. First, the SUR model can be rewritten as a single linear regression model
\[ y = X \beta + \varepsilon, \]
where $y = (y_1^\top,\ldots,y_m^\top)^\top$, $X = {\rm bdiag}(X_1,\ldots,X_m)$ a $nm \times p$ block diagonal matrix with $p = \sum_{j=1}^m p_j$, and $\beta = (\beta_1^\top,\ldots,\beta_m^\top)^\top$. For the error term $\varepsilon = (\varepsilon_1^\top,\ldots,\varepsilon_m^\top)^\top$ it then holds that ${\rm Cov}[\varepsilon] = \Sigma \kron I_n$. Secondly, the SUR model can be represented as a multivariate linear regression model
\[ Y = \tilde{X} \mathcal{B} + \mathcal{E}, \]
where $Y = (y_1,\ldots,y_m)$, $\tilde{X} = (X_1,\ldots,X_m)$, $\mathcal{B} = {\rm bdiag}(\beta_1,\ldots,\beta_m)$ and $\mathcal{E} = (\varepsilon_1,\ldots,\varepsilon_m)$. Equivalently, we can write the error matrix as $\mathcal{E} = (e_1,\ldots,e_n)^\top$ with $e_i = (\varepsilon_{i1},\ldots,\varepsilon_{im})^\top$ which satisfies ${\rm Cov}[e_i] = \Sigma$. Hence, the covariance of the error matrix $\mathcal{E}$ is given by ${\rm Cov}[\mathcal{E}] = \Sigma \kron I_n$.

It is well-known that ordinary least squares which ignores the correlation patterns across blocks may yield inefficient estimators. Generalized least squares (GLS) is a modification of least squares that can deal with any type of correlation, including contemporaneous correlation. For the SUR model, the GLS estimator takes the form
\begin{equation}
\label{beta_GLS}
\hat{\beta}_{\text{GLS}} = (X^\top (\Sigma^{-1} \kron I_n) X)^{-1} X^\top (\Sigma^{-1} \kron I_n) y.
\end{equation}
GLS coincides with the separate least squares estimates if $\sigma_{jk}$ for $j \neq k$, or if $X_1 = \ldots = X_m$. GLS is more efficient than least squares estimator~\citep{Zellner1962}, but in most situations the covariance $\Sigma$ needed in GLS is unknown. Feasible generalized least squares (FGLS) estimates the elements of $\Sigma$ by $\hat{\sigma}_{jk} = \hat{\varepsilon}_j^\top \hat{\varepsilon}_k / n$ where $\hat{\varepsilon}_j$ is the residual vector of the $j$th block obtained from ordinary least squares and then replaces $\Sigma$ in GLS by the resulting estimator $\hat{\Sigma}$. The finite-sample efficiency of FGLS is smaller than for GLS, although the asymptotic efficiency of both methods is identical. Note that FGLS can be repeated iteratively.

Alternatively, maximum likelihood estimators (MLE) can be considered~\citep[see][]{Srivastava1987}. Assuming that the disturbances are normally distributed, the log-likelihood of the SUR model is given by
\begin{equation}
\label{logLikelihood}
l(\beta,\Sigma|X,y) = - \frac{mn}{2} \ln(2 \pi) - \frac{n}{2} \ln(\abs{\Sigma}) - \frac{1}{2} (y - X \beta)^\top (\Sigma^{-1} \kron I_n) (y - X \beta).
\end{equation}
Maximizing this log-likelihood with respect to $(\beta,\Sigma)$ yields the estimators $(\hat{\beta}_{\text{MLE}},\hat{\Sigma}_{\text{MLE}})$ which are the solutions of the following equations
\begin{equation}
\begin{gathered}
\label{MLE}
\hat{\beta}_{\text{MLE}} = (X^\top (\hat{\Sigma}_{\text{MLE}}^{-1} \kron I_n) X)^{-1} X^\top (\hat{\Sigma}_{\text{MLE}}^{-1} \kron I_n) y \\
\hat{\Sigma}_{\text{MLE}} = (Y - \tilde{X} \hat{\mathcal{B}}_{\text{MLE}})^\top (Y - \tilde{X} \hat{\mathcal{B}}_{\text{MLE}}) / n
\end{gathered}
\end{equation}
with $\hat{\mathcal{B}}_{\text{MLE}}$ the block diagonal form of $\hat{\beta}_{\text{MLE}}$. Hence, the maximum likelihood estimators correspond to the fully iterated FGLS estimators.

It is well-known that outliers in the data (observations which deviate from the majority of the data) can severely influence classical estimators such as LS, MLE and their modifications. Hence, FGLS and MLE are expected to yield non-robust estimates. Robust M-estimators for the SUR model have been proposed, but these estimators lack affine equivariance~\citep{Koenker1990}. \citet{Bilodeau2000} have introduced robust and affine equivariant S-estimators. Recently,~\citet{Hubert2017} developed an efficient algorithm for these estimators. Despite its remarkable robustness properties, S-estimators can have a low efficiency, which makes them less suitable for inference. Therefore, we introduce MM-estimators for the SUR model which can combine high robustness with a high efficiency. To obtain efficient and powerful robust tests, we also introduce an efficient MM-estimator of the error scale based on the residuals of the MM-estimates.

Asymptotic theory can be used to draw inference corresponding to the MM-estimates in the SUR model. However, these asymptotic results rely on assumptions that are hard to verify in practice. The bootstrap~\citep{Efron1979} offers an alternative approach that does not require strict assumptions. However, the standard bootstrap lacks speed and robustness. Therefore, the fast and robust bootstrap (FRB) procedure of~\citet{Salibian2002} is adapted to the SUR setting. The FRB can be used to construct confidence intervals~\citep{Salibian2006,Salibian2002} as well as to develop hypothesis tests~\citep{Salibian2005,Salibian2016,VanAelst2011}. In particular, one of our main goals is to develop a robust test for diagonality of the covariance matrix $\Sigma$ to evaluate the need for using a SUR model.

To set the scene, MM-estimators for the SUR model are introduced in Section~\ref{Robust Estimators for the SUR Model} as an extension of S-estimators. Section~\ref{Fast and Robust Bootstrap} focuses on the fast and robust bootstrap procedure to develop robust inference. In Section~\ref{Robust Tests for the Regression Parameters} the MM-estimator of scale is introduced and hypothesis tests concerning the regression coefficients are studied. In Section~\ref{Robust Test for Diagonality of the Covariance Matrix} we investigate a robust procedure to test for diagonality of the covariance matrix $\Sigma$, i.e., to test whether a SUR model is really needed. The finite-sample performance of the FRB inference procedures is investigated by simulation in Section~\ref{Finite-Sample Performance}. Section~\ref{Example: Grunfeld Data} illustrates the robust inference on a real data example from economics and Section~\ref{Conclusion} concludes. The supplementary material includes properties of MM-estimators and the proposed test statistics, and contains some extra results on robust confidence intervals.


\section{Robust Estimators for the SUR Model}
\label{Robust Estimators for the SUR Model}


\subsection{S-estimators}

\noindent 
We first introduce S-estimators for the SUR model as proposed by~\citet{Bilodeau2000}. Consider so-called $\rho$-functions which satisfy the following conditions:
\begin{enumerate}
\item[(C1)] $\rho$ is symmetric, twice continuously differentiable and satisfies $\rho(0)=0$
\item[(C2)] $\rho$ is strictly increasing on $[0,c]$ and constant on $[c,\infty[$ for some $c>0$.
\end{enumerate}
The most popular family of $\rho$-functions is the class of Tukey bisquare $\rho$-functions given by $\rho(u) = {\rm min} (u^2/2 - u^4/2 c^2 + u^6/6 c^4,c^2/6)$ where $c>0$ is a tuning parameter.

\begin{definition}
\label{Sestimator}
Let $(X_j,y_j) \in \mathbb{R}^{n \times (p_j + 1)}$ for $j=1,\ldots,m$ and let $\rho_0$ be a $\rho$-function with parameter $c_0$ in (C2). Then, the S-estimators of the SUR model $(\tilde{\mathcal{B}},\tilde{\Sigma})$ are the solutions that minimize $\abs{C}$ subject to the condition
\[ \frac{1}{n} \sum_{i=1}^n \rho_0 \left( \sqrt{e_i(B)^\top C^{-1} e_i(B)} \right) = \delta_0, \]
where the minimization is over all $B = {\rm bdiag}(b_1,\ldots,b_m) \in \mathbb{R}^{p \times m}$ and $C \in \text{PDS}(m)$ with PDS$(m)$ the set of positive definite and symmetric matrices of dimension $m \times m$. The determinant of $C$ is denoted by $|C|$ and $e_i(B)^\top$ represents the $i$th row of the residual matrix $Y - \tilde{X} B$.
\end{definition}

The constant $\delta_0$ can be chosen as $\delta_0 = {\rm E}_F[\rho_0(\norm{e})]$ to obtain a consistent estimator at an assumed error distribution $F$. Usually, the errors are assumed to follow a normal distribution with mean zero and then we can take $F \sim \mathcal{N}_m(0,I_m)$. As before, the regression coefficient estimates in the matrix $\tilde{\mathcal{B}}$ can also be collected in the vector $\tilde{\beta} = (\tilde{\beta}_1^\top,\ldots,\tilde{\beta}_m^\top)^\top$.

The first-order conditions corresponding to the above minimization problem yield the following fixed-point equations for S-estimators
\begin{equation}
\begin{gathered}
\label{EstimatingEquationS}
\tilde{\beta} = (X^\top (\tilde{\Sigma}^{-1} \kron \tilde{D}) X)^{-1} X^\top (\tilde{\Sigma}^{-1} \kron \tilde{D}) y \\
\tilde{\Sigma} = m (Y - \tilde{X} \tilde{\mathcal{B}})^\top \tilde{D} (Y - \tilde{X} \tilde{\mathcal{B}}) \left( \sum_{i=1}^n v_0(\tilde{d}_i) \right)^{-1}
\end{gathered}
\end{equation}
with diagonal matrix $\tilde{D} = {\rm diag}(w_0(\tilde{d}_1),\ldots,w_0(\tilde{d}_n))$ where $\tilde{d}_i^2 = e_i(\tilde{\mathcal{B}})^\top \tilde{\Sigma}^{-1} e_i(\tilde{\mathcal{B}}),$ $w_0(u) = \psi_0(u)/u$, $\psi_0(u) = \rho_0'(u)$ and $v_0(u) = \psi_0(u) u - \rho_0(u) + \delta_0$. Note the similarities with the GLS in~\eqref{beta_GLS} and the MLE in~\eqref{MLE}. The factor $w_0(\tilde{d}_i)$ can be interpreted as the weight that the estimator gives to the $i$th observation. A small (large) residual distance $\tilde{d}_i$ leads to a large (small) weight $w_0(\tilde{d}_i)$. The smaller the weight of an observation, the smaller its contribution to the SUR fit. To compute the S-estimates efficiently,~\citet{Hubert2017} developed the fastSUR algorithm based on the ideas of~\citet{Salibian2006b}.

The breakdown point of an estimator is the smallest fraction of the data that needs to be contaminated in order to drive the bias of the estimator to infinity. S-estimators with a bounded loss function, as we consider here, have a positive breakdown point~\citep{Lopuhaa1991,VanAelst2005}. Their asymptotic breakdown point equals $\varepsilon^* = \delta_0/\rho_0(c_0)$. The constant $\delta_0$ has been fixed to guarantee consistency, but the parameter $c_0$ can be tuned to obtain any desired breakdown point $0 < \varepsilon^* \leq 0.5$. Hence, S-estimators can attain the maximal breakdown point of $50\%$. S-estimators with a smaller value of $c_0$ downweight observations more heavily and correspond to a higher breakdown point.

S-estimators satisfy the first-order conditions of M-estimators~\citep[see][]{Huber1981}, so they are asymptotically normal. However, the choice of the tuning parameter $c_0$ involves a trade-off between breakdown point (robustness) and efficiency at the central model~\citep{Bilodeau2000}. For this reason, S-estimators are less adequate for robust inference. MM-estimators~\citep{Yohai1987} avoid this trade-off by computing an efficient M-estimator starting from a highly robust S-estimator~\citep[see, e.g.,][]{Kudraszow2011,Tatsuoka2000,VanAelst2013}. We now introduce MM-estimators for the SUR model.


\subsection{MM-estimators}

\noindent Let $\tilde{\Sigma}$ denote the S-estimator of covariance in Definition~\ref{Sestimator}. Decompose $\tilde{\Sigma}$ into a scale component $\tilde{\sigma}$ and a shape matrix $\tilde{\Gamma}$ such that $\tilde{\Sigma} = \tilde{\sigma}^2 \tilde{\Gamma}$ with $\abs{\tilde{\Gamma}}=1$.

\begin{definition}
\label{MMestimator}
Let $(X_j,y_j) \in \mathbb{R}^{n \times (p_j + 1)}$ for $j=1,\ldots,m$ and let $\rho_1$ be a $\rho$-function with parameter $c_1$ in (C2). Given the S-scale $\tilde{\sigma}$, MM-estimators of the SUR model $(\hat{\mathcal{B}},\hat{\Gamma})$ minimize
\[ \frac{1}{n} \sum_{i=1}^n \rho_1 \left( \frac{\sqrt{e_i(B)^\top G^{-1} e_i(B)}}{\tilde{\sigma}} \right), \]
over all $B = {\rm bdiag}(b_1,\ldots,b_m) \in \mathbb{R}^{p \times m}$ and $G \in \text{PDS}(m)$ with $\abs{G} = 1$. The MM-estimator for covariance is defined as $\hat{\Sigma} = \tilde{\sigma}^2 \hat{\Gamma}$.
\end{definition}

As before, the MM-estimator of the regression coefficients $\hat{\mathcal{B}}$ can also be written in vector form $\hat{\beta}=(\hat{\beta}_1^\top,\ldots,\hat{\beta}_m^\top)^\top$. Similarly as for S-estimators, the first-order conditions corresponding to the above minimization problem yield a set of fixed-point equations:
\begin{equation}
\begin{gathered}
\label{EstimatingEquationMM}
\hat{\beta} = (X^\top (\hat{\Sigma}^{-1} \kron D) X)^{-1} X^\top (\hat{\Sigma}^{-1} \kron D) y \\
\hat{\Sigma} = m (Y - \tilde{X} \hat{\mathcal{B}})^\top D (Y - \tilde{X} \hat{\mathcal{B}}) \left( \sum_{i=1}^n \psi_1(d_i) d_i \right)^{-1}
\end{gathered}
\end{equation}
with $D = {\rm diag}(w_1(d_1),\ldots,w_1(d_n))$ where $d_i^2 = e_i(\hat{\mathcal{B}})^\top \hat{\Sigma}^{-1} e_i(\hat{\mathcal{B}})$, $w_1(u) = \psi_1(u)/u$ and $\psi_1(u) = \rho_1'(u)$. Starting from the initial S-estimates, the MM-estimates are calculated easily by iterating these estimating equations until convergence.

MM-estimators inherit the breakdown point of the initial S-estimators. Hence, they can attain the maximal breakdown point if initial high-breakdown point S-estimators are used. Moreover, since MM-estimators also satisfy the first-order conditions of M-estimators, they are asymptotically normal. In the supplementary material it is shown that the asymptotic efficiency of $\hat{\beta}$ does not depend on the $\rho$-function $\rho_0$ of the initial S-estimator. Therefore, the breakdown point and the efficiency of MM-estimators can be tuned independently. That is, the tuning constant $c_0$ in $\rho_0$ can be chosen to obtain an S-scale estimator with maximal breakdown point, while the constant $c_1 (> c_0)$ in $\rho_1$ is tuned to attain a desired efficiency, e.g., $90\%$, at the central model with normal errors. Note that while MM-estimators have maximal breakdown point, there is some loss of robustness because the bias due to contamination is generally higher as compared to S-estimators~\citep[see, e.g.,][]{Berrendero2007}.


\section{Fast and Robust Bootstrap}
\label{Fast and Robust Bootstrap}

\noindent The asymptotic distribution of MM-estimators can be used to obtain inference for the parameters in the SUR model based on their MM-estimates. However, these asymptotic results are only reasonable for sufficiently large samples and rely on the assumption of elliptically symmetric errors which does not necessarily hold in practice. The bootstrap offers an alternative approach that requires less assumptions. Unfortunately, for robust estimators the standard bootstrap procedure lacks speed and robustness. The standard bootstrap is computer intensive because many bootstrap replicates are needed and the fastSUR algorithm is itself already computationally intensive. Moreover, classical bootstrap does not yield robust inference results. Indeed, due to the resampling with replacement, the proportion of outlying observations varies among bootstrap samples. Some bootstrap samples thus contain a majority of outliers, resulting in breakdown of the estimator. These estimates affect the bootstrap distribution leading to unreliable inference. Therefore, we use the fast and robust bootstrap introduced by~\citet{Salibian2002} and generalized in e.g., \citet{Salibian2006} and~\citet{Peremans2017}.

Consider an estimator $\hat{\theta}$ of a parameter $\theta$ that satisfies the fixed-point equations $g(\hat{\theta})=\hat{\theta}$ where the function $g$ depends on the given sample. For a bootstrap sample it equivalently holds that $g^*(\hat{\theta}^*)=\hat{\theta}^*$. Now, consider $g^*(\hat{\theta})$ as a first-step approximation of the bootstrap estimate $\hat{\theta}^*$. These first-step approximations underestimate the variability of the bootstrap distribution since the starting value is the same for all bootstrap approximations. To remedy this deficiency a linear correction factor can be derived from a Taylor expansion of $g^*(\hat{\theta}^*)$. This yields the fast and robust bootstrap (FRB) estimator, given by
\[ \hat{\theta}^{R*} = \hat{\theta} + (I - \nabla g(\hat{\theta}) )^{-1} (g^*(\hat{\theta}) - \hat{\theta}), \]
with $\nabla g(\hat{\theta})$ the gradient of $g$ evaluated at $\hat{\theta}$. Consistency of $\hat{\theta}^{R*}$ has been discussed in detail by~\citet{Salibian2002,Salibian2006}. The FRB estimator is computationally much more efficient because the first-step approximations are easy to compute and the linear correction term needs to be calculated only once, since it depends only on the original sample. Moreover, for a robust estimator the fixed-point equations usually correspond to a weighted version of the corresponding equations for the non-robust MLE or generalized least squares estimator. The weights in the equations downweight outlying observations. In such case, the FRB estimator is robust because no matter how many times an outlying observation appears in a bootstrap sample, it receives the same low weight as in the original sample since the weights depend on the estimate $\hat{\theta}$ corresponding to the original sample.

To apply the FRB to the S and MM-estimators for the SUR model, we rewrite the estimating equations of S-estimators in~\eqref{EstimatingEquationS} as
\begin{align*}
g_4(\tilde{\beta},\tilde{\Sigma}) \quad &= \quad (X^\top ( \tilde{\Sigma}^{-1} \kron \tilde{D}) X)^{-1} X^\top ( \tilde{\Sigma}^{-1} \kron \tilde{D}) y \\
g_3(\tilde{\beta},\tilde{\Sigma}) \quad &= \quad m (Y - \tilde{X} \tilde{\mathcal{B}})^\top \tilde{D} (Y - \tilde{X} \tilde{\mathcal{B}}) \left( \sum_{i=1}^n v_0(\tilde{d}_i) \right)^{-1}
\end{align*}
where $\tilde{D} = {\rm diag}(w_0(\tilde{d}_1),\ldots,w_0(\tilde{d}_n))$, $\tilde{d}_i^2 = \tilde{e}_i( \tilde{\mathcal{B}})^\top \tilde{\Sigma}^{-1} \tilde{e}_i( \tilde{\mathcal{B}})$. Similarly, we rewrite the estimating equations~\eqref{EstimatingEquationMM} of MM-estimators as
\begin{align*}
g_1(\hat{\beta},\hat{\Gamma},\tilde{\Sigma}) \quad &= \quad (X^\top ( \hat{\Gamma}^{-1} \kron D) X)^{-1} X^\top ( \hat{\Gamma}^{-1} \kron D) y \\
g_2(\hat{\beta},\hat{\Gamma},\tilde{\Sigma}) \quad &= \quad \phi ( (Y - \tilde{X} \hat{\mathcal{B}})^\top D (Y - \tilde{X} \hat{\mathcal{B}}) )
\end{align*}
where $D = {\rm diag}(w_1(d_1),\ldots,w_1(d_n))$, $d_i^2 = \abs{\tilde{\Sigma}}^{-1/m} e_i(\hat{\mathcal{B}})^\top \hat{\Gamma}^{-1} e_i(\hat{\mathcal{B}})$, and $\phi(A) = \abs{A}^{-1/m} A$ for an $m \times m$ matrix $A$. Now, let $\hat{\theta}=(\hat{\beta}^\top,{\rm vec}(\hat{\Gamma})^\top,{\rm vec}(\tilde{\Sigma})^\top,\tilde{\beta}^\top)^\top$ be the vector which combines the S and MM-estimates for the SUR model and let
\begin{equation}
\label{FRBgfunctionMMestimator}
g(\hat{\theta}) = (g_1(\hat{\beta},\hat{\Gamma},\tilde{\Sigma})^\top,g_2(\hat{\beta},\hat{\Gamma},\tilde{\Sigma})^\top,g_3(\tilde{\beta},\tilde{\Sigma})^\top,g_4(\tilde{\beta},\tilde{\Sigma}))^\top.
\end{equation}
Then, we have that $g(\hat{\theta})=\hat{\theta}$. Expressions for the partial derivatives in $\nabla g$ can be found in the supplementary material.

Based on the FRB estimates $\hat{\theta}^{R*}$ confidence intervals for the model parameters can be constructed by using standard bootstrap techniques. This is shown in more detail in the supplementary material. In the next sections we construct robust test procedures for the SUR model and show how FRB can be used to estimate their null distribution.


\section{Robust Tests for the Regression Parameters}
\label{Robust Tests for the Regression Parameters}

\noindent Consider the following general null and alternative hypothesis with respect to the regression parameters in the SUR model
\begin{equation}
\label{HypothesisTestbeta}
H_0: R \beta = q \quad \text{vs} \quad H_1: R \beta \neq q,
\end{equation}
for some $R \in \mathbb{R}^{r \times p}$ and $q \in \mathbb{R}^r$. Here $r \leq p$ represents the number of linear restrictions on the regression parameters under the null hypothesis. For example, for $R = (0,\ldots,0,1)$ and $q = 0$ the null hypothesis simplifies to $\beta_{p_mm} = 0$. Note that the null hypothesis can restrict regression parameters of different blocks, e.g., $H_0: \beta_{11} = \beta_{12}$.

For maximum likelihood estimation, the standard test statistic is the well-known likelihood-ratio statistic. With the log-likelihood in~\eqref{logLikelihood} it is given by
\[ \Lambda_{\text{MLE}} = -n \ln \left( \frac{\abs{\hat{\Sigma}_{\text{MLE}}}}{\abs{\hat{\Sigma}_{\text{MLE},r}}} \right), \] 
where $\hat{\Sigma}_{\text{MLE}}$ is the MLE in the full model and $\hat{\Sigma}_{\text{MLE},r}$ the MLE in the restricted model under the null hypothesis. Under the null hypothesis the test statistic is asymptotically chi-squared distributed with $r$ degrees of freedom. See, e.g., \citet{Henningsen2007} for more details on standard test statistics (such as Wald and F-statistics) in SUR models.

A robust likelihood-ratio type test statistic corresponding to MM-estimators can be obtained by using the plug-in principle. Let $\hat{\Sigma}$ denote the unrestricted scatter MM-estimator and $\hat{\Sigma}_r$ the restricted MM-estimator. Then, the robust likelihood-ratio statistic becomes
\begin{equation}
\label{RobustLikelihoodRatioS}
\Lambda_{\text{S}} = -n \ln \left( \frac{\abs{\hat{\Sigma}}}{\abs{\hat{\Sigma}_r}} \right) = - 2nm \ln \left( \frac{\tilde{\sigma}}{\tilde{\sigma}_r} \right),
\end{equation}
with $\tilde{\sigma}$ and $\tilde{\sigma}_r$ the scale S-estimators of the full and null model, respectively. Similarly to $\Lambda_{\text{MLE}}$, the test statistic $\Lambda_{\text{S}}$ is nonnegative, since $\tilde{\sigma} \leq \tilde{\sigma}_r$ by definition of the S-estimators.

The test statistic $\Lambda_{\text{S}}$ in~\eqref{RobustLikelihoodRatioS} only depends on S-scale estimators. Hence, the low efficiency of S-estimators may affect the efficiency of tests based on $\Lambda_{\text{S}}$. In the linear regression context,~\citet{VanAelst2013b} recently introduced an efficient MM-scale estimator corresponding to regression MM-estimators. Analogously, we propose to update the S-estimator of scale $\tilde{\sigma}$ in the SUR model by a more efficient M-scale $\hat{\sigma}$, defined as
\[ \hat{\sigma} = \tilde{\sigma} \sqrt{\frac{1}{n \delta_1} \sum_{i=1}^n \rho_1 \left( \frac{\sqrt{e_i(\hat{\mathcal{B}})^\top \hat{\Gamma}^{-1} e_i(\hat{\mathcal{B}})}}{\tilde{\sigma}} \right)}. \]
Similarly to $\delta_0$, the constant $\delta_1$ can be chosen as $\delta_1 = {\rm E}_F [\rho_1(\norm{e})]$ to obtain a consistent estimator at the assumed error distribution $F$, e.g., $F \sim \mathcal{N}_m(0,I_m)$. The likelihood-ratio type test statistic corresponding to this MM-scale estimator is then defined as
\begin{equation}
\label{RobustLikelihoodRatioM}
\Lambda_{\text{MM}} = - 2nm \ln \left( \frac{\hat{\sigma}}{\hat{\sigma}_r} \right).
\end{equation}

Results on the asymptotic distribution and influence function of these test statistics are provided in the supplementary material. Since the asymptotic distribution is only useful for sufficiently large samples, we consider FRB as an alternative to estimate the null distribution of the test statistics. However, since likelihood-ratio type test statistics converge at a higher rate than the estimators themselves, a standard application of FRB leads to an inconsistent estimate of the null distribution of the test statistic~\citep{VanAelst2011}. To overcome this issue, the test statistic $\Lambda_{\text{S}}$ in~\eqref{RobustLikelihoodRatioS} is rewritten as
\begin{equation}
\label{RobustLikelihoodRatioS_consistent}
\Lambda_{\text{S}} = - 2nm \ln \left( \frac{\tilde{s}(\tilde{\mathcal{B}},\tilde{\Gamma})}{\tilde{s}(\tilde{\mathcal{B}}_r,\tilde{\Gamma}_r)} \right),
\end{equation}
where $(\tilde{\mathcal{B}},\tilde{\Gamma})$ and $(\tilde{\mathcal{B}}_r,\tilde{\Gamma}_r)$ are the S-estimators in the full and null model respectively and where $\tilde{s}(B,G)$ is the multivariate M-estimator of scale corresponding to a given $B \in \mathbb{R}^{p \times m}$ and $G \in \text{PDS}(m)$ with $\abs{G}=1$. That is, $\tilde{s}(B,G)$ is the solution of
\begin{equation}
\label{MultivariateMscale}
\frac{1}{n} \sum_{i=1}^n \rho_0 \left( \frac{\sqrt{e_i(B)^\top G^{-1} e_i(B)}}{\tilde{s}(B,G)} \right) = \delta_0.
\end{equation} 
Similarly, the MM-based test statistic $\Lambda_{\text{MM}}$ in~\eqref{RobustLikelihoodRatioM} is rewritten as
\begin{equation}
\label{RobustLikelihoodRatioM_consistent}
\Lambda_{\text{MM}} = - 2nm \ln \left( \frac{\hat{s}(\tilde{\mathcal{B}},\tilde{\Gamma},\hat{\mathcal{B}},\hat{\Gamma})}{\hat{s}(\tilde{\mathcal{B}}_r,\tilde{\Gamma}_r,\hat{\mathcal{B}}_r,\hat{\Gamma}_r)} \right),
\end{equation}
where
\begin{equation}
\label{GeneralizedMultivariateMscale}
\hat{s}(\tilde{\mathcal{B}},\tilde{\Gamma},\hat{\mathcal{B}},\hat{\Gamma}) = \tilde{s}(\tilde{\mathcal{B}},\tilde{\Gamma}) \sqrt{\frac{1}{n \delta_1} \sum_{i=1}^n \rho_1 \left( \frac{\sqrt{e_i(\hat{\mathcal{B}})^\top \hat{\Gamma}^{-1} e_i(\hat{\mathcal{B}})}}{\tilde{s}(\tilde{\mathcal{B}},\tilde{\Gamma})} \right)}.
\end{equation}

Let $\hat{\theta} = (\hat{\beta}^\top,{\rm vec}(\hat{\Gamma})^\top,{\rm vec}(\tilde{\Gamma})^\top,\tilde{\beta}^\top)^\top$ contain the S and MM-estimators of the regression coefficients and error shape matrices for the full model and let $\hat{\theta}_r$ contain the corresponding estimators for the reduced model. Denote $\hat{\Theta} = (\hat{\theta},\hat{\theta}_r)$, then both test statistics can be written in the general form
\[ \Lambda_. = h(\hat{\Theta}), \]
where the dot in the subscript can be either S or MM and the function $h$ is determined by~\eqref{RobustLikelihoodRatioS_consistent}-\eqref{MultivariateMscale} or~\eqref{RobustLikelihoodRatioM_consistent}-\eqref{GeneralizedMultivariateMscale}, respectively. The FRB approximation for the null distribution of this test statistic then consists of the values
\[ \Lambda_.^{R*} = h^{*}(\hat{\Theta}^{R*}), \]
where $\hat{\Theta}^{R*} = (\hat{\theta}^{R*},\hat{\theta}_r^{R*})$ are the FRB approximations for the regression and shape estimates in the bootstrap samples. It can be checked that the function $h$ satisfies the condition
\begin{equation}
\label{consistencycondition}
\nabla h(\hat{\Theta}) = o_p(1),
\end{equation}
so the partial derivatives of $h$ vanish asymptotically. This condition guarantees that the FRB procedure consistently estimates the null distribution of the test statistic, as shown in~\citet{VanAelst2011}. Note that the FRB procedure for hypothesis tests is computationally less efficient than for the construction of confidence intervals (see supplementary material) because the S-scales of the full and null model have to be computed by an iterative procedure for each of the bootstrap samples. However, the increase in computation time is almost negligible compared to the time needed by the standard (non-robust) bootstrap for these robust estimators.

Bootstrapping a test statistic to estimate its null distribution requires that the bootstrap samples follow the null hypothesis, even when this hypothesis does not hold in the original data. Therefore, we first construct null data that approximately satisfy the null hypothesis, regardless of the hypothesis that holds in the original data. According to~\citet{Salibian2016}, for the linear constraints in~\eqref{HypothesisTestbeta} null data for $\Lambda_{\text{MM}}$ can be constructed as
\[ (\tilde{X}^{(0)},Y^{(0)}) = (\tilde{X},\tilde{X} \hat{\mathcal{B}}_r + E), \]
with $E = Y - \tilde{X} \hat{\mathcal{B}}$ the residuals in the full model. Bootstrap samples are now generated by sampling with replacement from the null data $(\tilde{X}^{(0)},Y^{(0)})$. Let $(\hat{\mathcal{B}}^{(0)},\hat{\Sigma}^{(0)})$ denote the MM-estimates for the null data in the full model and let $(\hat{\mathcal{B}}_r^{(0)},\hat{\Sigma}_r^{(0)})$ denote the MM-estimates for the null data in the restricted model. Due to affine equivariance we have that $(\hat{\mathcal{B}}^{(0)},\hat{\Sigma}^{(0)}) = (\hat{\mathcal{B}}_r,\hat{\Sigma})$, so these estimates can be obtained without extra computations. However, the estimates for the reduced model cannot be derived from equivariance properties and need to be computed from the transformed data. Similarly, null data can be constructed for $\Lambda_{\text{S}}$. Finally, when $N$ FRB recalculated values $\Lambda_.^{R*}$ of the test statistic have been calculated based on the null data, then the corresponding FRB p-value is given by
\begin{equation}
\label{p-value}
\text{p-value}= \frac{(\# \Lambda_.^{R*} > \Lambda_.) + 1}{N + 2},
\end{equation}
where $\Lambda_.$ is the value of the test statistic at the original sample.


\section{Robust Test for Diagonality of the Covariance Matrix}
\label{Robust Test for Diagonality of the Covariance Matrix}

\noindent The key feature of the SUR model is the existence of contemporaneous correlation, corresponding to a non-diagonal covariance matrix $\Sigma$. If the covariance matrix is diagonal the SUR model simplifies to $m$ unrelated regression models. Therefore, by testing for diagonality of $\Sigma$ the necessity of a SUR model is evaluated.

Consider the following hypotheses
\begin{equation}
\label{HypothesisTestSigma}
H_0: \Sigma \text{ is diagonal } \quad \text{vs} \quad H_1: \Sigma \text{ is not diagonal}.
\end{equation}
A popular diagonality test for the standard SUR model is the Breusch-Pagan test~\citep{Breusch1980} which is based on the Lagrange multiplier idea~\citep{Baltagi2008}. It measures the total sum of squared correlations:
\[ \text{LM}_{\text{MLE}} = n \sum_{j<k} r_{jk}^2, \]
with $r_{jk}$ the elements of the sample correlation matrix of the residual vectors $\hat{\varepsilon}_j$, $j=1,\ldots,m$. Here, each $\hat{\varepsilon}_j$ is the residual vector corresponding to a single-equation LS fit in block $j$. Under the null hypothesis $\text{LM}_{\text{MLE}}$ is asymptotically chi-squared distributed with $m(m-1)/2$ degrees of freedom. Evidently, the LS based Breusch-Pagan test is vulnerable to outliers in the data. Therefore, we introduce robust Breusch-Pagan type tests.

Contrary to the classical estimators, the S and MM-estimators in a SUR model do not simplify to their univariate analogues under the null hypothesis. However, to calculate the restricted estimates the S and MM-estimators and corresponding fastSUR algorithm can be adapted such that the equations for the off-diagonal elements of the covariance matrix are excluded. For example, in case of MM-estimators the estimating equations become
\begin{equation*}
\begin{gathered}
\hat{\beta}_r = (X^\top (\hat{\Sigma}_r^{-1} \kron D) X)^{-1} X^\top (\hat{\Sigma}_r^{-1} \kron D) y \\
\hat{\sigma}_{r,jj} = m \left( \sum_{i=1}^n w_1(d_i) e_{ij}^2(\hat{\mathcal{B}}_r) \right) \left( \sum_{i=1}^n \psi_1(d_i) d_i \right)^{-1}
\end{gathered}
\end{equation*}
for $j=1,\ldots,m$ and with $D = {\rm diag}(w_1(d_1),\ldots,w_1(d_n))$ where $d_i^2 = e_i(\hat{\mathcal{B}}_r)^\top \hat{\Sigma}_r^{-1} e_i(\hat{\mathcal{B}}_r)$. The restricted covariance matrix estimates $\tilde{\Sigma}_r$ and $\hat{\Sigma}_r$ under $H_0$ then become diagonal matrices as needed. Since the tuning constants of the $\rho$-functions are kept fixed, the reduced estimators $(\tilde{\beta}_r,\tilde{\Sigma}_r,\hat{\beta}_r,\hat{\Sigma}_r)$ also have the same breakdown-point and efficiency level as their counterparts in the full model. Moreover, the multivariate structure is not lost, i.e., we still obtain a single weight for each observation across all blocks.

Based on the restricted estimators, we now estimate the correlation between the errors of block $j$ and $k$ as
\[ r_{jk} = \frac{\sum_{i=1}^n w_1(d_i) e_{ij}(\hat{\mathcal{B}}_r) e_{ik}(\hat{\mathcal{B}}_r)}{\sqrt{\left( \sum_{i=1}^n w_1(d_i) e_{ij}^2(\hat{\mathcal{B}}_r) \right) \left( \sum_{i=1}^n w_1(d_i) e_{ik}^2(\hat{\mathcal{B}}_r) \right)}}, \]
with $d_i^2 = e_i( \hat{\mathcal{B}}_r)^\top \hat{\Sigma}_r^{-1} e_i( \hat{\mathcal{B}}_r)$. Based on these correlation estimates we propose a robust Breusch-Pagan test statistic:
\begin{equation}
\label{RobustLM}
\text{LM}_{\text{MM}} = n \sum_{j<k} r_{jk}^2.
\end{equation}
Note that $\text{LM}_{\text{MM}}$ is nonnegative. Similarly, a robust Breusch-Pagan test based on S-estimators, denoted by $\text{LM}_{\text{S}}$, can be defined as well, but it will not benefit from the gain in efficiency of MM-estimators.

From their asymptotic chi-squared distribution (see the supplementary material) non-robust p-values may be derived. Alternatively, FRB can again be used to estimate the null distribution of the test statistics. Note that the robust Breusch-Pagan test statistic only requires the estimates in the restricted model as can be expected for a Lagrange multiplier test. Let $\hat{\theta}_r$ denote the vector that collects all S and MM-estimators in the restricted model. Based on the FRB approximations $\hat{\theta}_r^{R*}$, bootstrap replications for the null distribution of $\text{LM}_{\text{MM}}$ can be generated as
\[ \text{LM}_{\text{MM}}^{R*} = n \sum_{j<k} (r_{jk}^{R*})^2, \]
with
\[ r_{jk}^{R*} = \frac{\sum_{i=1}^n w_1(d_i^{R*}) e_{ij}(\hat{\mathcal{B}}_r^{R*}) e_{ik}(\hat{\mathcal{B}}_r^{R*})}{\sqrt{\left( \sum_{i=1}^n w_1(d_i^{R*}) e_{ij}^2(\hat{\mathcal{B}}_r^{R*}) \right) \left( \sum_{i=1}^n w_1(d_i^{R*}) e_{ik}^2(\hat{\mathcal{B}}_r^{R*}) \right)}}, \]
where $(d_i^{R*})^2 = e_i(\hat{\mathcal{B}}_r^{R*})^\top (\hat{\Sigma}_r^{R*})^{-1} e_i( \hat{\mathcal{B}}_r^{R*})$, and similarly for $\text{LM}_{\text{S}}$. It is straightforward to check that the consistency condition in~\eqref{consistencycondition} holds under $H_0$ for these test statistics, where $h$ is now defined through~\eqref{RobustLM}. Hence, the FRB procedure consistently estimates the null distribution of the test statistics.

To make sure that the bootstrap samples satisfy the null hypothesis, we generate bootstrap samples from the following transformed data
\begin{equation*}
(\tilde{X}^{(0)},Y^{(0)}) = (\tilde{X},\tilde{X} \hat{\mathcal{B}} + E \hat{\Sigma}^{-1/2}),
\end{equation*}
with $E = Y - \tilde{X} \hat{\mathcal{B}}$ the residuals in the full model. The residuals $E$ of the full SUR model are possibly correlated across blocks. By transforming these residuals with $\hat{\Sigma}^{-1/2}$, this correlation is removed and it can be expected that for the transformed data
\begin{equation}
\label{nullmatrix}
\hat{\Sigma}^{(0)} \approx I_m \quad \text{and} \quad \hat{\Sigma}_r^{(0)} \approx I_m,
\end{equation}
regardless of the hypothesis that holds in the original data. Note that in the SUR model we cannot rely on equivariance properties to obtain the identity matrix exactly because the model is only affine equivariant for transformations within blocks. However, extensive empirical investigation confirmed that~\eqref{nullmatrix} holds for the transformed data, and the corresponding value of the test statistic $\text{LM}_{\text{MM}}^{(0)}$ indeed becomes approximately zero. Similarly, null data can be created for $\text{LM}_{\text{S}}$ as well.


\section{Finite-Sample Performance}
\label{Finite-Sample Performance}

We now investigate by simulation the performance of FRB tests based on the robust likelihood-ratio test statistics $\Lambda_{\text{S}}$ and $\Lambda_{\text{MM}}$ and the robust Breusch-Pagan statistics $\text{LM}_{\text{S}}$ and $\text{LM}_{\text{MM}}$. The tests are performed at the $5\%$ significance level. We study both the efficiency of the tests under the null hypothesis and the power under the alternative as well as their robustness.

In the SUR model, bootstrap samples can be obtained by either case (row) resampling from the original sample $(\tilde{X},Y)$ or by resampling the $m$-dimensional residuals $e_i$, $i=1,\ldots,n$. While the results in the previous sections hold for both types of bootstrapping, in this paper we use case resampling which is a more nonparametric approach than the model based error resampling.

Consider first the following hypothesis test in a SUR model:
\begin{equation}
\label{HypothesisTestbeta_simulation}
H_0: \beta_{p_mm} = 0 \quad \text{vs} \quad H_1: \beta_{p_mm} \neq 0.
\end{equation}
To investigate the efficiency of the test procedures, data are simulated under the null hypothesis. Observations are generated according to a SUR model with three blocks ($m=3$) and two predictors (as well as an intercept) in each block. Hence, there are $p=9$ regression coefficients in the model. The predictor variables are generated independently from a standard normal distribution. The $p$-dimensional vector of regression coefficients equals $\beta = (1,\ldots,1,0)^\top$ such that the null hypothesis holds. The covariance matrix $\Sigma$ is taken to be a correlation matrix with all correlations equal to 0.5. The multivariate errors are generated from either $\mathcal{N}_m(0,\Sigma)$ or $t_m(0,\Sigma)$ (a multivariate elliptical t-distribution with mean zero and scatter $\Sigma$) with 3 degrees of freedom. To investigate the robustness of the procedure we also considered contaminated data. We have generated the worst possible type of outliers, namely bad leverage points, by replacing in each block all the regressors of the first 10\% or 30\% of the observations by uniform values between -10 and -5 and by adding to each of the corresponding original responses a value that is normally distributed with mean 20 and variance 1.

Robust S-estimators and MM-estimators with maximal breakdown point of 50\% are computed. The MM-estimator is tuned to have 90\% efficiency. The null distribution of both $\Lambda_{\text{S}}$ and $\Lambda_{\text{MM}}$ are estimated by FRB as explained in Section~\ref{Robust Tests for the Regression Parameters}, using $N=1000$ bootstrap samples. The corresponding p-values are obtained as in~\eqref{p-value}. For each simulation setting 1000 random samples are generated for sample sizes $n=25,50,75,100,150,200,250$ and 300 (recall that $n$ represents the number of observations per block). Figure~\ref{paper_simulationFRBsurLR_betatestH0} shows the empirical level of the two tests for both clean and contaminated data.
\begin{figure}[!ht]
\centering
\includegraphics[width=\textwidth, trim= 0mm 101.6mm 0mm 0mm, clip]{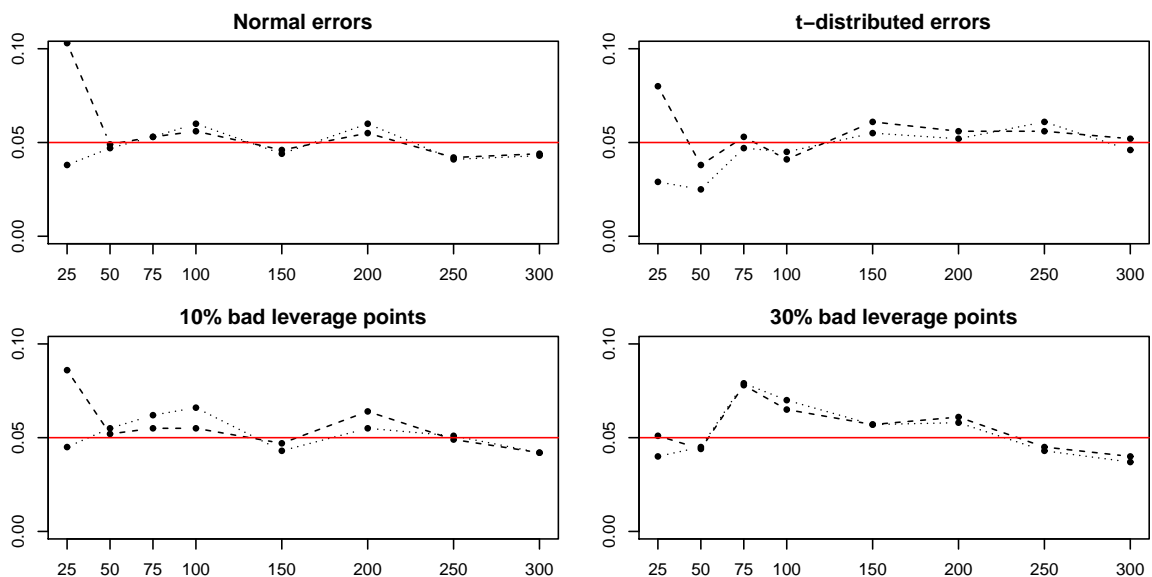}
\caption{Rejection rates of the hypothesis test in~\eqref{HypothesisTestbeta_simulation} based on the test statistics $\Lambda_{\text{S}}$ (dashed) and $\Lambda_{\text{MM}}$ (dotted). The solid (red) line represents the rejection level of 5\%.}
\label{paper_simulationFRBsurLR_betatestH0}
\end{figure}
It can be seen that the empirical levels are close to the 5\% nominal level in most cases. The difference between $\Lambda_{\text{S}}$ and $\Lambda_{\text{MM}}$ is mainly seen when the sample size is small. Indeed, for $n=25$, the test using $\Lambda_{\text{MM}}$ performs better than when $\Lambda_{\text{S}}$ is used. Note that outliers in the data only have a limited effect on the rejection rates, showing robustness of the level of the FRB tests.

To investigate the power of the robust tests, we have simulated data sets under the alternative hypothesis. In Figure~\ref{paper_simulationFRBsurLR_betatestH1} we show the power of the tests for samples of size $n=100$ with $\beta = (1,\ldots,1,d)^\top$ where $d$ ranges from 0 to 0.5 with step length 0.1.
\begin{figure}[!ht]
\centering
\includegraphics[width=\textwidth, trim= 0mm 152.4mm 0mm 0mm, clip]{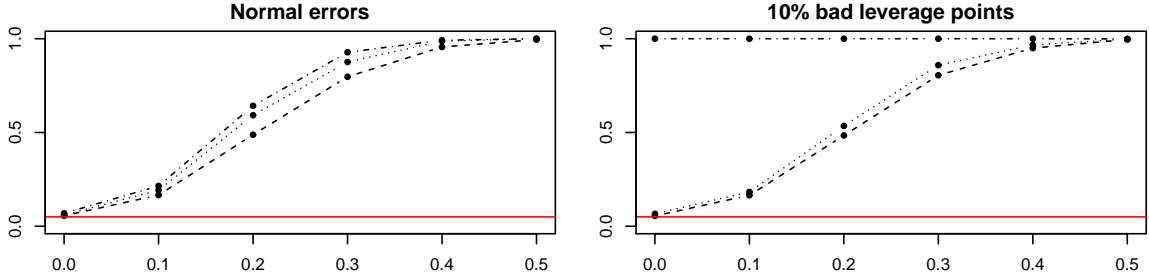}
\caption{Power curves of the hypothesis test in~\eqref{HypothesisTestbeta_simulation} based on the test statistics $\Lambda_{\text{S}}$ (dashed), $\Lambda_{\text{MM}}$ (dotted) and $\Lambda_{\text{MLE}}$ (dash-dotted). The solid (red) line represents the rejection level of 5\%.}
\label{paper_simulationFRBsurLR_betatestH1}
\end{figure}
From the left plot we see that the power increases quickly when $d$ becomes larger. The power of the robust tests is only slightly lower than for the classical test in the non-contaminated setting. Moreover, the power of the $\Lambda_{\text{MM}}$ test is (slightly) higher than for the $\Lambda_{\text{S}}$ test. The plot on the right shows that the classical test completely fails if the data is contaminated with 10\% of bad leverage points. On the other hand, the robust tests are not affected much by the contamination and yield similar power curves as in the case without contamination.

Let us now consider the test for diagonality of the covariance matrix in~\eqref{HypothesisTestSigma}. First, data are generated under the null hypothesis, i.e., data are simulated as in the previous section, but the multivariate errors are generated from either $\mathcal{N}_m(0,\Sigma)$ or from $t_m(0,\Sigma)$ with $\Sigma$ the identity matrix. The LM test statistic corresponding to both S and MM-estimators is computed. As before, 1000 data sets were generated for each setting. In Figure~\ref{paper_simulationFRBsurLM_diagtestH0} the rejection rates are plotted as a function of sample size for the four cases considered (normal errors, t-distributed errors, 10\% contamination and 30\% contamination).
\begin{figure}[!ht]
\centering
\includegraphics[width=\textwidth, trim= 0mm 101.6mm 0mm 0mm, clip]{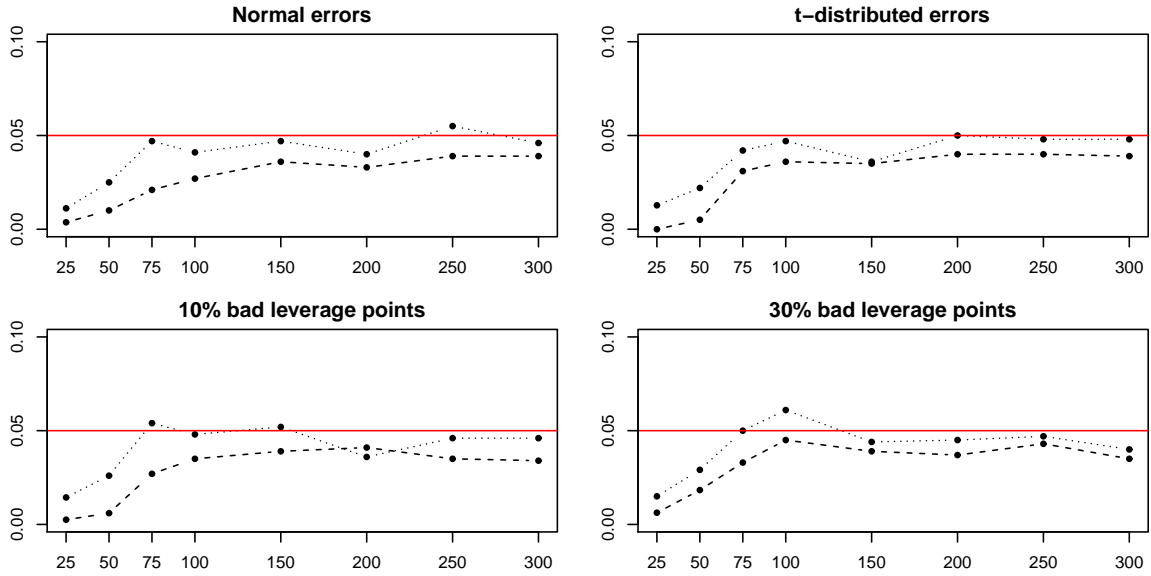}
\caption{Rejection rates of the hypothesis test in~\eqref{HypothesisTestSigma} based on the test statistics $\text{LM}_{\text{S}}$ (dashed) and $\text{LM}_{\text{MM}}$ (dotted). The solid (red) line represents the rejection level of 5\%.}
\label{paper_simulationFRBsurLM_diagtestH0}
\end{figure}
The rejection rates in the different cases behave similar. The lower efficiency of S-estimators becomes apparent as the empirical levels of $\text{LM}_{\text{S}}$ are lower in all (but one) cases. For small sample sizes the nominal level is clearly underestimated, but for MM-estimation the nominal level is already reached for $n \geq 75$. The efficiency of the tests is not much affected by heavy tailed errors or contamination which confirms their robustness under the null hypothesis.

To investigate the power of the test procedures, data were simulated under the alternative hypothesis as well. To this end, $\Sigma$ was set equal to an equicorrelation matrix with correlation $\tau$ taking values from 0 to 0.5 with step length 0.1 for the case $n=100$.
\begin{figure}[ht!]
\centering
\includegraphics[width=\textwidth, trim= 0mm 152.4mm 0mm 0mm, clip]{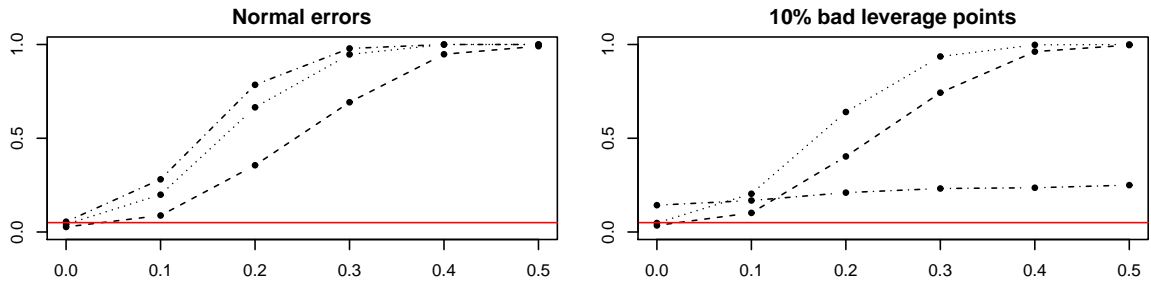}
\caption{Power curves of the hypothesis test in~\eqref{HypothesisTestSigma} based on the test statistics $\text{LM}_{\text{S}}$ (dashed), $\text{LM}_{\text{MM}}$ (dotted) and $\text{LM}_{\text{MLE}}$ (dash-dotted). The solid (red) line represents the rejection level of 5\%.}
\label{paper_simulationFRBsurLM_diagtestH1}
\end{figure}
The left plot in Figure~\ref{paper_simulationFRBsurLM_diagtestH1} shows the resulting power curves of the classical and robust Breusch-Pagan tests. We see that the test based on MM-estimators performs almost as well as the classical Breusch-Pagan test. For $\tau=0.4$ the empirical level of $\text{LM}_{\text{MM}}$ reaches almost one. The test based on S-estimators performs less well in this setting with $m=3$ blocks. However, we have noted that the performance of $\text{LM}_{\text{S}}$ increases with the number of blocks $m$ in the SUR model. For larger block sizes the difference with $\text{LM}_{\text{MM}}$ becomes negligible. The right plot in Figure~\ref{paper_simulationFRBsurLM_diagtestH1} shows that the classical Breusch-Pagan test cannot handle contamination, resulting in a drastic loss of power. On the other hand, the power of the robust tests is not affected much by the bad leverage points, resulting in power curves that are similar to the uncontaminated case. This setting where $\Sigma$ is an equicorrelation matrix can be considered to be a strong deviation from diagonality because the deviation is present in all covariance elements. Therefore, we also investigated the power of the diagonality test for other structures of $\Sigma$. It turns out that the comparison between the three tests remains the same for other settings. The power curves for the case where only one covariance deviates from zero are given in the supplementary material.


\section{Example: Grunfeld Data}
\label{Example: Grunfeld Data}

\noindent As an illustration we consider the well-known Grunfeld data (see, e.g., \citet{Bilodeau2000}). This dataset contains information on the annual gross investment of 10 large U.S.\ corporations for the period 1935-1954. The recorded response is the annual gross investment of each corporation (Investment). Two predictor variables have been measured as well, which are the value of outstanding shares at the beginning of the year (Shares) and the beginning-of-year real capital stock (Capital). One may expect that within the same year the activities of one corporation can affect the others. Hence, the SUR model seems to be appropriate. Unfortunately, the classical and robust estimators of the covariance matrix become singular when all 10 companies are considered. Therefore, we only focus on the measurements of three U.S.\ corporations: General Electric (GE), Westinghouse (W) and Diamond Match (DM). General Electric and Westinghouse are active in the same field of industry and thus their activities can highly influence each other. Since the interest is in modeling dependencies between the corporations within the same year, a SUR model with three blocks is considered. The model is given by
\begin{equation}
\label{Grunfeld_model}
\text{Investment}_{ij} = \beta_{0j} + \beta_{1j} \text{ Shares}_{ij} + \beta_{2j} \text{ Capital}_{ij} + \varepsilon_{ij},
\end{equation}
with ${\rm Cov}[\varepsilon_{ij},\varepsilon_{ik}]=\sigma_{jk}$ for $i=1,\ldots,20$ and $j,k=1,2,3$.

We consider inference corresponding to the standard MLE and robust MM-estimators. MM-estimates are obtained with 50\% breakdown point and a normal efficiency of 90\%. For the MLE, inference is obtained by using asymptotic results and standard bootstrap. For MM-estimators, robust inference is based on the asymptotic results as well as on FRB using $N=1000$ bootstrap samples generated by case resampling. Given the small sample size, we may expect that the bootstrap inference is more reliable than the asymptotic inference according the simulation results in the previous section.

Table~\ref{Grunfeld_beta_se} contains the estimates for the regression coefficients and corresponding standard errors (between brackets) based on bootstrap for the SUR model in~\eqref{Grunfeld_model}.
\begin{table}[!ht]
\renewcommand{\arraystretch}{1.25}
\begin{center}
\begin{tabular}{l c c c c c c}
\hline
\textbf{Corporation} & \multicolumn{3}{c}{\textbf{MLE}} & \multicolumn{3}{c}{\textbf{MM-estimator}} \\
 & Intercept & Shares & Capital & Intercept & Shares & Capital \\
\hline
\multirow{2}{*}{GE} &  -42.270  &  0.049  &  0.122  &  -30.661 &  0.033  &  0.152  \\
                    &  (27.559) & (0.016) & (0.034) & (26.679) & (0.014) & (0.026) \\
\hline
\multirow{2}{*}{W}  &   -3.684  &  0.067  &  0.018  &  -6.320  &  0.059  &  0.117  \\
                    &   (8.293) & (0.016) & (0.074) & (10.779) & (0.022) & (0.102) \\
\hline
\multirow{2}{*}{DM} &   -0.716  &  0.016  &  0.453  &  -0.855  &  0.002  &  0.614  \\
                    &   (1.394) & (0.022) & (0.144) &  (0.608) & (0.009) & (0.093) \\
\hline
\end{tabular}
\end{center}
\caption{Estimated regression coefficients and bootstrap standard errors (between brackets) for the MLE and MM-estimator applied to the SUR model for the Grunfeld data. Standard errors have been obtained by classical bootstrap (MLE) or FRB (MM-estimates).}
\label{Grunfeld_beta_se}
\end{table}
We can clearly see that there are differences between the estimates of both procedures. Focusing on the slope estimates, we see that the MM-estimator yields larger effects of Capital (beginning-of-year real capital stock) and smaller effects of Shares (value of outstanding shares at beginning of the year) on annual gross investments than the MLE. The largest differences can be seen in the estimates $\hat{\beta}_{22}$, $\hat{\beta}_{13}$, and $\hat{\beta}_{23}$ and their standard errors. The estimates for the scatter matrix $\Sigma$ and corresponding correlation matrix $R$ are given by
\[ \hat{\Sigma}_{\text{MLE}} =
\left[
\begin{array}{r@{}l  r@{}l  r@{}l}
 784&.2 & 224&.2 & 19&.4 \\
    &   &  97&.8 &  6&.5 \\
    &   &    &   &  1&.0 \\
\end{array} \right],
\quad
R_{\text{MLE}} = 
\left[
\begin{array}{r  r  r}
1 & 0.81 & 0.69 \\
  & 1    & 0.65 \\
  &      & 1    \\
\end{array} \right],
\]
and
\[ \hat{\Sigma}_{\text{MM}} = \left[
\begin{array}{r@{}l  r@{}l  r@{}l}
520&.9 & 194&.6 & 6&.1 \\
   &   & 110&.1 & 2&.6 \\
   &   &    &   & 0&.2 \\
\end{array} \right],
\quad
R_{\text{MM}} =
\left[
\begin{array}{r r  r}
1 & 0.81 & 0.56 \\
  & 1    & 0.52 \\
  &      & 1    \\
\end{array} \right], 
\]
respectively. The robust covariance estimates are generally smaller than the classical estimates. Both estimators find large correlations between the errors of the different blocks. The largest correlation occurs between the first two blocks, which correspond to the equations of General Electric and Westinghouse.

Since there are several differences between the non-robust MLE and the robust MM-estimates, we investigate the data for the presence of outliers. Outliers can be detected by constructing a multivariate diagnostic plot as in~\citet{Hubert2017}. This plot displays the residual distances of the observations versus the robust distance of its predictors. Based on the SUR estimates the residual distances are computed as
\[ d_i = \sqrt{e_i(\hat{\mathcal{B}}_{\text{MM}})^\top \hat{\Sigma}_{\text{MM}}^{-1} e_i(\hat{\mathcal{B}}_{\text{MM}})}. \]
Similarly, to measure how far an observations lies from the majority in the predictor space, robust distances can be calculated as
\[ \text{RD}_i = \sqrt{(\tilde{X}_i - \hat{m}_{\text{MM}})^\top \hat{C}_{\text{MM}}^{-1} (\tilde{X}_i - \hat{m}_{\text{MM}})}, \]
with $\tilde{X}_i$ the $i$th row of $\tilde{X}$ and where $\hat{m}_{\text{MM}}$ and $\hat{C}_{\text{MM}}$ are MM-estimates of the location and scatter of $\tilde{X}$~\citep{Tatsuoka2000}. Note that contributions of intercept terms have been removed from $\tilde{X}$ so that only the actual predictors are taken into account. For non-outlying observations with normal errors, the squared residual distances are asymptotically chi-squared distributed with $m$ degrees of freedom as usual. Therefore, a horizontal line at cut-off value $\sqrt{\chi_{m,0.975}^2}$ (the square root of the 0.975 quantile of a chi-squared distribution with $m$ degrees of freedom) is added to the plot to flag outliers. Observations that exceed this cut-off are considered to be outliers. Similarly, if the predictors of the regular observations are approximately normally distributed, then asymptotically the squared robust distances are approximately chi-squared distributed with $p$ degrees of freedom. Therefore, we add a vertical line to the plot at cut-off value $\sqrt{\chi_{p,0.975}^2}$ to identify outliers in the predictor space, i.e., leverage points. An observation is called a vertical outlier if its residual distance exceeds the cut-off but it is not outlying in the predictor space. If the observation is also outlying in the predictor space, it is called a bad leverage point. Observations with small residual distance which are outlying in the predictor space are called good leverage points because they still follow the SUR model. Similarly, a diagnostic plot can be constructed based on the initial S-estimates for the SUR model or even based on the MLE, although the latter will not reliably identify outliers due to the non-robustness of the estimates.

Multivariate diagnostic plots corresponding to our analysis of the Grunfeld data are shown in Figure~\ref{paper_Grunfeld_diagnostic}, based on both the MLE and MM-estimates.
\begin{figure}[!ht]
\centering
\includegraphics[width=\textwidth]{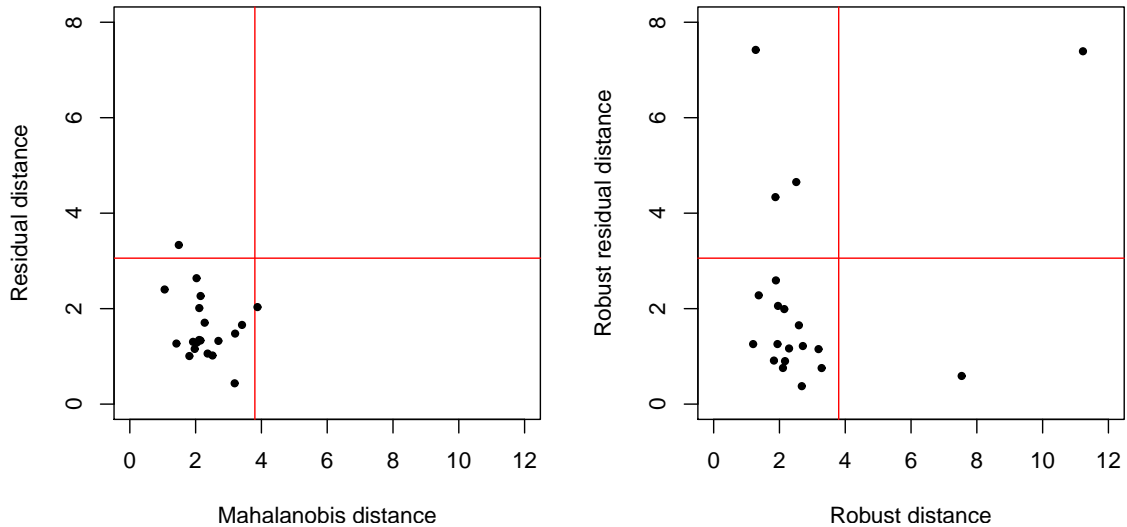}
\caption{Multivariate diagnostic plots based on the classical estimates (left panel) and robust estimates (right panel) for three companies in the Grunfeld data}
\label{paper_Grunfeld_diagnostic}
\end{figure}
The diagnostic plot corresponding to the classical non-robust estimates does not reveal any clear outliers. It seems that all observations follow the SUR model. However, outliers may have affected the estimates to the extent that the outliers are masked. Therefore, we consider the robust diagnostic plot corresponding to the MM-estimates. This plot indeed shows a different picture. Three vertical outliers and one bad leverage point are identified, as well as one good leverage point. The three vertical outliers correspond to the years 1946, 1947 and 1948, while the bad leverage point corresponds to the year 1954. Further exploration of the data indicates that the three vertical outliers are mainly due to exceptionally high investments in those three post World War II years. For the final year 1954, the measurements for all variables are rather extreme, most likely due to the postwar booming economy, which explains why this year is flagged as a bad leverage point in the robust analysis. These four outliers may potentially influence the inference results based on MLE, leading to misleading conclusions. To verify the effect of the outliers on the MLE estimates of the parameters, we also calculated the MLE estimates based on the data without the outliers. The results (not shown) confirmed that the outliers and especially the bad leverage point affect the MLE estimates, because without these outliers the MLE estimates highly resemble the MM-estimates in Table~\ref{Grunfeld_beta_se}.

The large correlation estimates between the errors of the different blocks already suggested that these correlations should not be ignored, and thus that the SUR model is indispensable. We can now formally test whether it is indeed necessary to use the SUR model. Therefore, we apply the diagonality test in Section~\ref{Robust Test for Diagonality of the Covariance Matrix} to test the hypotheses in~\eqref{HypothesisTestSigma}. Table~\ref{Grunfeld_LM} shows the results for the Breusch-Pagan test as well as our robust Breusch-Pagan test. The table contains the values of both test statistics, as well as the corresponding asymptotic p-values and bootstrap p-values. The proportionality constant for the asymptotic chi-squared distribution is estimated by using the empirical distribution to calculate the expected value.
\begin{table}[ht!]
\renewcommand{\arraystretch}{1.25}
\begin{center}
\begin{tabular}{l c c c}
\hline
\textbf{Estimator} & \textbf{LM} & \textbf{AS p-value} & \textbf{B p-value} \\
\hline
\multirow{1}{*}{MLE} & 23.482 & 0.001 & 0.003 \\
\multirow{1}{*}{MM}  & 14.825 & 0.003 & 0.019 \\
\hline
\end{tabular}
\end{center}
\caption{Results of the classical and robust Breusch-Pagan test for the hypothesis test in~\eqref{HypothesisTestSigma} using the Grunfeld data.}
\label{Grunfeld_LM}
\end{table}
We immediately see that at the $5\%$ significance level, the null hypothesis of diagonality is rejected in all cases. Hence, the outliers in this example do not affect the MLE estimates in such a way that the covariance structure of the SUR model is completely hidden.

From an econometric point of view it can now be interesting to investigate whether the predictors Shares and Capital have the same effect on investments for the two energy companies General Electric and Westinghouse. Hence, we test
\begin{equation}
\label{Grunfeld_HypothesisTestbeta}
H_0: \beta_{11} = \beta_{12} \text{ and } \beta_{21} = \beta_{22} \quad \text{vs} \quad H_1: \beta_{11} \neq \beta_{12} \text{ or } \beta_{21} \neq \beta_{22}.
\end{equation}
Table~\ref{Grunfeld_LR} contains the values of the likelihood-ratio statistics and corresponding asymptotic and bootstrap p-values.
\begin{table}[ht!]
\renewcommand{\arraystretch}{1.25}
\begin{center}
\begin{tabular}{l c c c}
\hline
\textbf{Estimator} & \textbf{$\Lambda$} & \textbf{AS p-value} & \textbf{B p-value} \\
\hline
\multirow{1}{*}{MLE} & 6.728 & 0.035 & 0.168 \\
\multirow{1}{*}{MM}  & 7.255 & 0.057 & 0.086 \\
\hline
\end{tabular}
\end{center}
\caption{Classical and robust test results for the hypothesis test in~\eqref{Grunfeld_HypothesisTestbeta} using the Grunfeld data.}
\label{Grunfeld_LR}
\end{table}
If we consider a $5\%$ significance level, then the conclusion is not completely clear for the MLE. The commonly used asymptotic p-value does reject the null hypothesis, but based on the bootstrap p-value we cannot reject the null hypothesis anymore. On the other hand, the robust test yields asymptotic and bootstrap p-values that lie closer together and which do not reject the null hypothesis. Hence, the presence of outliers does not affect the outcome of the robust hypothesis test while it seems to have caused instability for the classical test based on the MLE. Indeed, if we remove the bad leverage point, then the asymptotic p-value corresponding to the MLE already increases to $0.061$ which is in line with the p-value based on the MM-estimator for the full data set.


\section{Conclusion}
\label{Conclusion}

\noindent In this paper we have introduced MM-estimators for the SUR model as an extension of S-estimators. MM-estimators combine high robustness (breakdown point) with high efficiency at the central model. Based on these MM-estimators robust inference for the SUR model has been developed based on the FRB principle. We considered likelihood ratio type statistics to test the existence of linear restrictions among the regression coefficients. While MM-estimators update the S-estimates of the regression coefficients and shape matrix, they do not automatically update the S-scale estimate. However, it turns out that more accurate and powerful tests are obtained if a more efficient MM-scale estimator is used.

An important question is whether it is necessary to use a joint SUR model rather than individual linear regression models for each of the blocks. To evaluate the need for a SUR model we proposed a robust alternative for the well-known Breusch-Pagan test. The FRB was used again to obtain a highly reliable test for diagonality of the covariance matrix, i.e., for existence of contemporaneous correlation among the errors in the different blocks of the SUR model.


\section*{Acknowledgments}

\noindent This research has been partially supported by grant C16/15/068 of International Funds KU Leuven and the CRoNoS COST Action IC1408. The computational resources and services used in this work were provided by the VSC (Flemish Supercomputer Center), funded by the Research Foundation - Flanders (FWO) and the Flemish Government - department EWI.


\section*{Supplementary Material}

\noindent In the supplementary material we introduce functionals corresponding to MM-estimators and discuss important properties of these MM-functionals such as equivariance, influence function and asymptotic variance. Also influence functions and asymptotic distributions are derived for the proposed robust test statistics. Power curves are included for a situation which is less deviating from diagonality than the equicorrelation matrix. Furthermore, we construct bootstrap confidence intervals based on FRB and evaluate their performance in a simulation study. In addition, we illustrate these confidence intervals on Grunfeld data. The appendix also contains expressions for the partial derivatives required in the FRB procedure, a verification of the consistency conditions for the robust test on regression coefficients, and the proofs of the theorems.


\bibliographystyle{apalike}
\bibliography{ms_references}

\newpage
\spacingset{1} 


\if0\blind
{
  \title{\bf Robust Inference for Seemingly Unrelated Regression Models: Supplementary Material}
  \author{Kris Peremans and Stefan Van Aelst \hspace{.2cm} \\
    Department of Mathematics, KU Leuven, 3001 Leuven, Belgium \\
  }
  \maketitle
} \fi

\if1\blind
{
  \vspace*{0.35cm}
  \begin{center}
    {\LARGE \bf Robust Inference for Seemingly Unrelated Regression Models: Supplementary Material\par}
  \end{center}
  \vspace*{3.5cm}
} \fi

\bigskip
\begin{abstract} 
\noindent In the supplementary material we introduce functionals corresponding to MM-estimators and discuss important properties of these MM-functionals such as equivariance, influence function and asymptotic variance. Also influence functions and asymptotic distributions are derived for the proposed robust test statistics. Power curves are included for a situation which is less deviating from diagonality than the equicorrelation matrix. Furthermore, we construct bootstrap confidence intervals based on fast and robust bootstrap and evaluate their performance in a simulation study. In addition, we illustrate these confidence intervals on Grunfeld data. The appendix also contains expressions for the partial derivatives required in the fast and robust bootstrap procedure, a verification of the consistency conditions for the robust test on regression coefficients, and the proofs of the theorems. 
\end{abstract}

\noindent KEYWORDS: 
Asymptotic normal efficiency; Influence function; Fast and robust bootstrap; Robust confidence interval
\vfill

\newpage
\spacingset{1.45} 


\section{Properties of MM-estimators}
\label{Properties of MM-estimators}

\noindent We investigate the properties of MM-estimators in more detail. To this end, we first introduce MM-functionals corresponding to the MM-estimators introduced in the manuscript. We state equivariance properties of these MM-functionals and investigate their robustness and efficiency by deriving their influence function and asymptotic variance. We present results for both the estimator of the regression coefficients $\hat{\beta}$ and the estimator of the scatter $\hat{\Sigma}$.


\subsection{Functionals}

\noindent Functional versions of S and MM-estimators for the SUR model can be defined as follows.

\begin{definition}
Let $H: \mathbb{R}^{p+m} \longrightarrow \mathbb{R}$ be the distribution function of $(\tilde{X}^\top,Y^\top)^\top$ and let $\rho_0$ be a $\rho$-function as before. Then, the S-functionals of the SUR model $(\tilde{\mathcal{B}}(H),\tilde{\Sigma}(H))$ are the solutions that minimize $\abs{C}$ subject to the condition
\[ {\rm E}_H \left[ \rho_0 \left( \sqrt{e(B)^\top C^{-1} e(B)} \right) \right] = \delta_0, \]
over all $B = {\rm bdiag}(b_1,\ldots,b_m) \in \mathbb{R}^{p \times m}$ and $C \in \text{PDS}(m)$ with $e(B)=Y - B^\top \tilde{X}$.
\end{definition}

To define the MM-functionals, we again decompose the scatter matrix functional into a scale and a shape component, i.e., $\tilde{\Sigma}(H) = \tilde{\sigma}^2(H) \tilde{\Gamma}(H)$ such that $\abs{\tilde{\Gamma}(H)}=1$.

\begin{definition}
Let $H: \mathbb{R}^{p+m} \longrightarrow \mathbb{R}$ be the distribution function of $(\tilde{X}^\top,Y^\top)^\top$ and let $\rho_1$ be a $\rho$-function as before. Given the S-scale functional $\tilde{\sigma}(H)$, the MM-functionals of the SUR model $(\hat{\mathcal{B}}(H),\hat{\Gamma}(H))$ minimize
\[ {\rm E}_H \left[ \rho_1 \left( \frac{ \sqrt{e(B)^\top G^{-1} e(B)}}{\tilde{\sigma}(H)} \right) \right], \]
over all $B = {\rm bdiag}(b_1,\ldots,b_m) \in \mathbb{R}^{p \times m}$ and $G \in \text{PDS}(m)$ with $\abs{G} = 1$. The MM-functional for covariance is defined as $\hat{\Sigma}(H) = \tilde{\sigma}^2(H) \hat{\Gamma}(H)$.
\end{definition}

Note that the S and MM-estimators can be obtained by the choice $H=\hat{H}_n$, the empirical distribution function corresponding to the data.


\subsection{Equivariance}

\noindent Similarly as for S-estimators~\citep{Bilodeau2000}, it can easily be shown that the MM-functionals in the SUR model are equivariant under affine transformations of the regressors, regression transformations and blockwise scale transformations of the responses.

For ease of notation, let us write the MM-functionals as $\hat{\beta}(X,Y)$ and $\hat{\Sigma}(X,Y)$ with $X = {\rm bdiag} (X_1^\top, \ldots, X_m^\top)$ with $X_j \in \mathbb{R}^{p_j}$. Then, the MM-functionals satisfy the following equivariance properties:
\begin{enumerate}[(a)]
\item Affine equivariance of regressors:
\[ \hat{\beta}(X A,Y) = A^{-1} \hat{\beta}(X,Y) \quad \text{and} \quad \hat{\Sigma}(XA,Y) = \hat{\Sigma}(X,Y), \]
with $A= {\rm bdiag}(A_1,\ldots,A_m)$ where the blocks $A_j$ are of size $p_j \times p_j$.
\item Regression equivariance:
\[ \hat{\beta}(X,Y + Xa) = \hat{\beta}(X,Y) + a \quad \text{and} \quad \hat{\Sigma}(X,Y + Xa) = \hat{\Sigma}(X,Y), \]
for any $a \in \mathbb{R}^p$.
\item Scale equivariance of responses:
\[ \hat{\beta}(X,AY) = \tilde{A} \hat{\beta}(X,Y) \quad \text{and} \quad \hat{\Sigma}(X,AY) = A \hat{\Sigma}(X,Y) A, \]
for any diagonal matrix $A={\rm diag}(a_{11},\ldots,a_{mm})$ with diagonal matrix $\tilde{A} = {\rm diag}(a_{11},\ldots,a_{11},\ldots,a_{mm},\ldots,a_{mm})$ in which each diagonal element $a_{jj}$ of $A$ is repeated $p_j$ times.
\end{enumerate}


\subsection{Influence Function}

\noindent We now derive the influence functions of the MM-functionals introduced above. Since MM-functionals reduce to S-functionals when $\rho_1=\rho_0$, we only have to consider influence functions for MM-functionals. While the breakdown point is a global measure of robustness, the influence function is a local measure of robustness. The influence function of a functional $T$ measures the effect on $T$ of an infinitesimal amount of contamination at a point $z=(\tilde{x}^\top,y^\top)^\top \in \mathbb{R}^{p+m}$. Consider the contaminated distribution
\[ H_{\epsilon,\Delta_z} = (1 - \epsilon) H + \epsilon \Delta_z, \]
with $\Delta_z$ the point mass distribution at $z$ and $0<\epsilon<1$. Then, the influence function of $T$ is defined as
\[ \text{IF}(z;T,H) = \lim_{\epsilon \to 0} \frac{T(H_{\epsilon,\Delta_z}) - T(H)}{\epsilon} = \diff(,\epsilon) \left( T(H_{\epsilon,\Delta_z}) \right) \Big{|}_{\epsilon = 0}. \]

To derive the influence function, we consider the SUR model
\[ Y = B^\top \tilde{X} + \mathcal{E} = X \beta + \mathcal{E} , \]
where the $p$-dimensional vector $\tilde{X}$ has distribution $K$ and is independent of the $m$-dimensional error variable $\mathcal{E}$. We assume that $\mathcal{E}$ follows a unimodal elliptically symmetric distribution $F_\Sigma$ with density
\[ f_\Sigma(u)= \abs{\Sigma}^{-1/2} g(u^\top\Sigma^{-1}u), \]
where $\Sigma \in \text{PDS}(m)$ and the function $g$ has a strictly negative derivative. The error distribution is thus symmetric around the origin. Let $H_{\beta,\Sigma}$ denote the resulting distribution of $Z=(\tilde{X}^\top,Y^\top)^\top$. The following theorem gives the influence functions of the regression and scatter MM-functionals for model distributions $H_{\beta,\Sigma}$.

\begin{theorem}
\label{TheoremInfluenceFunctions}
If $Z=(\tilde{X}^\top,Y^\top)^\top$ has model distribution $H_{\beta,\Sigma}$ as defined above, then the influence functions of the MM-estimators for the SUR model are given by
\begin{equation}
\label{InfluenceFunctionbeta}
\text{IF}(z;\hat{\beta},H_{\beta,\Sigma}) = \frac{1}{\eta_1} w_1(\norm{e}_{\Sigma}) {\rm E}_K [X^\top \Sigma^{-1} X]^{-1} x^\top \Sigma^{-1} e,
\end{equation}
and
\[ \text{IF}(z;\hat{\Sigma},H_{\beta,\Sigma}) = \frac{m}{\pi_1} \psi_1(\norm{e}_{\Sigma}) \norm{e}_{\Sigma} \left( \frac{e e^\top}{\norm{e}_{\Sigma}^2} - \frac{1}{m} \Sigma \right) + \frac{2}{\gamma_0} (\rho_0(\norm{e}_{\Sigma}) - \delta_0) \hspace{0.5mm} \Sigma, \]
with $e=y-x\beta$ and where we use the notation $\norm{a}_{C}^2 = a^\top C^{-1} a$ for $a \in \mathbb{R}^{m}$ and $C \in \text{PDS}(m)$. With $F:= F_\Sigma $ the constants are given by
\begin{align}
\label{eta_1}
\eta_1 &= {\rm E}_F \left[ \left( 1-\frac{1}{m} \right) w_1(\norm{\mathcal{E}}_{\Sigma}) + \frac{1}{m} \psi_1'(\norm{\mathcal{E}}_{\Sigma}) \right], \\
\label{pi_1}
\pi_1 &= \frac{1}{m+2} {\rm E}_F [(m+1) \psi_1(\norm{\mathcal{E}}_{\Sigma}) \norm{\mathcal{E}}_{\Sigma} + \psi_1'(\norm{\mathcal{E}}_{\Sigma}) \norm{\mathcal{E}}_{\Sigma}^2], \\
\label{gamma_0}
\gamma_0 &= {\rm E}_F [\psi_0(\norm{\mathcal{E}}_{\Sigma}) \norm{\mathcal{E}}_{\Sigma}].
\end{align}
\end{theorem}

Note that the influence function of the regression functional $\hat{\beta}$ is bounded in $e$ but unbounded in $x$. Hence, contamination in the direction of the response has a bounded influence on $\hat{\beta}$. The effect becomes zero for far away outliers because the weight function $w_1(\norm{e}_{\Sigma})$ becomes zero for large values of its argument. On the other hand, contamination in the predictor space can have an infinitely large effect on the estimator, but only if the corresponding residual is sufficiently small. This means that the point is a good leverage point since it does not deviate from the SUR model. Moreover, the influence function of the scatter functional $\hat{\Sigma}$ only depends on $e$ and is bounded. Hence, contamination in the predictor space does not affect the scatter functional while the effect of contamination in the response remains bounded.


\subsection{Asymptotic Variance}

\noindent Following~\citet{Hampel1986}, the asymptotic variance of a functional $T$ is obtained by
\[ \text{ASV} (T,H) = {\rm E}_H [\text{IF}(z;T,H) \text{IF}(z;T,H)^\top]. \]
By using the expressions for the influence functions in Theorem~\ref{TheoremInfluenceFunctions} we immediately obtain the asymptotic variances of the MM-estimators for the SUR model in Theorem~\ref{TheoremAsymptoticVariance} below. We use the notation $K_m$ for the commutation matrix of size $m^2 \times m^2$ such that $K_m {\rm vec} (A) = {\rm vec} (A)^\top$ for any matrix $A \in \mathbb{R}^{m \times m}$. Note that vec denotes the vector operator, stacking all columns of its matrix argument into one vector.

\begin{theorem}
\label{TheoremAsymptoticVariance}
If $Z=(\tilde{X}^\top,Y^\top)^\top$ has model distribution $H_{\beta,\Sigma}$, then the asymptotic variances of the MM-estimators for the SUR model are given by
\begin{equation}
\label{AsymptoticVariancebeta}
\text{ASV} (\hat{\beta},H_{\beta,\Sigma}) = \frac{\alpha_1}{m \eta_1^2} {\rm E}_K [X^\top \Sigma^{-1} X]^{-1},
\end{equation}
and
\[ \text{ASV} (\hat{\Sigma},H_{\beta,\Sigma}) = \sigma_1 (I_{m^2} + K_m) (\Sigma \kron \Sigma) + \sigma_2 {\rm vec} (\Sigma) {\rm vec} (\Sigma)^\top. \]
With $F:= F_\Sigma $ as before, the constants $\alpha_1$, $\sigma_1$ and $\sigma_2$ are equal to
\begin{align}
\label{alpha_1}
\alpha_1 &= {\rm E}_F [\psi_1^2(\norm{\mathcal{E}}_{\Sigma})], \\
\nonumber \sigma_1 &= \frac{m}{\pi_1^2 (m+2)} {\rm E}_F [\psi_1^2(\norm{\mathcal{E}}_{\Sigma}) \norm{\mathcal{E}}_{\Sigma}^2], \\
\nonumber \sigma_2 &= \frac{4}{\gamma_0^2} {\rm E}_F [(\rho_0(\norm{\mathcal{E}}_{\Sigma})-\delta_0)^2] - \frac{2}{m} \sigma_1,
\end{align}
and $\eta_1$, $\pi_1$ and $\gamma_0$ are given by~\eqref{eta_1}, \eqref{pi_1} and~\eqref{gamma_0}, respectively.
\end{theorem}

In case of S-estimators ($\rho_1=\rho_0$), these expressions correspond to the asymptotic variances of S-estimators in~\citet{Bilodeau2000}. Moreover, the asymptotic variance of the scatter $\hat{\Sigma}$ coincides with that in~\citet{Lopuhaa1989} and~\citet{Salibian2006}.

The asymptotic relative efficiency (ARE) for the regression coefficients $\hat{\beta}$, relative to the MLE $\hat{\beta}_{\text{MLE}}$, becomes
\[ \text{ARE}(\hat{\beta},H_{\beta,\Sigma}) = \frac{\text{ASV}(\hat{\beta}_{\text{MLE}},H_{\beta,\Sigma})}{\text{ASV}(\hat{\beta},H_{\beta,\Sigma})} = \frac{m \eta_1^2}{\alpha_1}. \]
Note that the ARE does not depend on the number of predictors $p$ in the SUR model nor on the distribution of $\tilde{X}$, but only depends on the number of blocks $m$ in the model and the distribution of the errors. Moreover, it can immediately be seen that the ARE of the MM-estimator $\hat{\beta}$ does not depend on the initial loss function $\rho_0$ for the S-estimator, but only depends on the loss function $\rho_1$. Hence, the constant $c_1$ in $\rho_1$ can indeed be tuned to guarantee a desired efficiency at the central model, independently of the breakdown point which is determined by the constant $c_0$ in $\rho_0$.


\section{Asymptotic Results of the Proposed Test Statistics}
\label{Asymptotic Results of the Proposed Test Statistics}

Furthermore, we present some asymptotic results of the robust test statistics $\Lambda_.$ and $\text{LM}_.$ (see Sections~\ref{Robust Tests for the Regression Parameters} and~\ref{Robust Test for Diagonality of the Covariance Matrix} of the manuscript respectively).


\subsection{Robust Tests for the Regression Parameters}

\noindent Under the null hypothesis in~\eqref{HypothesisTestbeta} the asymptotic distributions of the test statistics $\Lambda_{\text{S}}$ and $\Lambda_{\text{MM}}$ are proportional to a chi-squared distribution with $r$ degrees of freedom. Denote
\begin{align*}
\eta_\ell &= {\rm E}_F \left[ \left( 1-\frac{1}{m} \right) w_\ell(\norm{\mathcal{E}}_{\Sigma}) + \frac{1}{m} \psi_\ell'(\norm{\mathcal{E}}_{\Sigma}) \right], \\
\gamma_\ell &= {\rm E}_F [\psi_\ell(\norm{\mathcal{E}}_{\Sigma}) (\norm{\mathcal{E}}_{\Sigma})], \\
\alpha_\ell &= {\rm E}_F [\psi_\ell^2(\norm{\mathcal{E}}_{\Sigma})],
\end{align*}
for $\ell=0,1$ and with $F = F_{\Sigma}$. Remark that the constants $\eta_1$, $\gamma_0$ and $\alpha_1$ are already defined in~\eqref{eta_1},~\eqref{gamma_0} and~\eqref{alpha_1} respectively. Then, we have the following result.

\begin{theorem}
\label{TheoremAsymptoticDistributionLambda}
Let $Z=(\tilde{X}^\top,Y^\top)^\top$ have model distribution $H_{\beta,\Sigma}$. Under the null hypothesis $H_0: R \beta = q$ it holds that
\[ \Lambda_{\text{S}} \quad \overset{d}{\longrightarrow} \quad \frac{\alpha_0}{\eta_0 \gamma_0} \hspace{1mm} \chi_r^2, \]
and
\[ \Lambda_{\text{MM}} \quad \overset{d}{\longrightarrow} \quad \frac{\alpha_1}{\eta_1 \gamma_1} \hspace{1mm} \chi_r^2. \]
\end{theorem}

These asymptotic null distributions can be used to obtain p-values corresponding to the test statistics in the finite-sample case. However, this standard approach requires a sufficiently large sample size and also accurate estimates of the expectations in the proportionality factors to obtain reliable results.

Robustness of these test statistics is investigated through their influence functions. The (first-order) influence function of these test statistics equals zero. Therefore, we consider their second-order influence function~\citep{Croux2008}, which is defined as
\[ \text{IF2}(z;T,H) = \diff(^2,\epsilon^2) \left( T(H_{\epsilon,\Delta_z}) \right) \Big{|}_{\epsilon = 0}. \]
Boundedness of this influence function guarantees stability of the asymptotic level and power of the asymptotic test in presence of contamination~\citep{Heritier1994}. The next theorem yields the second-order influence functions of $\Lambda_{\text{S}}$ and $\Lambda_{\text{MM}}$ at model distribution $H_{\beta,\Sigma}$ under the null hypothesis $H_0$.

\begin{theorem}
\label{TheoremInfluenceFunctionsLambda}
If $Z=(\tilde{X}^\top,Y^\top)^\top$ has model distribution $H := H_{\beta,\Sigma}$ and if $H_0$ is true, then the second-order influence functions of $\Lambda_{\text{S}}$ and $\Lambda_{\text{MM}}$ are given by
\[
\text{IF2}(z;\Lambda_{\text{S}},H) = - \frac{2m \eta_0}{\gamma_0} (R \hspace{0.5mm} \text{IF}(z;\tilde{\beta},H) - q)^\top \left( R \hspace{0.5mm} {\rm E}_K [X^\top \Sigma^{-1} X]^{-1} R^\top \right)^{-1} (R \hspace{0.5mm} \text{IF}(z;\tilde{\beta},H) - q),\]
and
\[ \text{IF2}(z;\Lambda_{\text{MM}},H) = \left( 1 - \frac{\gamma_1}{2\delta_1} \right) \text{IF2}(z;\Lambda_{\text{S}},H) \\
- \frac{m \eta_1}{\delta_1} (R \hspace{0.5mm} \text{IF}(z;\hat{\beta},H) - q)^\top \left( R \hspace{0.5mm} {\rm E}_K [X^\top \Sigma^{-1} X]^{-1} R^\top \right)^{-1} (R \hspace{0.5mm} \text{IF}(z;\hat{\beta},H) - q). \]
\end{theorem}

The second-order influence functions in Theorem~\ref{TheoremInfluenceFunctionsLambda} are unbounded in $x$ but bounded in $e$. The redescending nature of the functions $w_0$ and $w_1$ guarantees that contamination in the response does not affect the test statistics when $\norm{e}_{\Sigma}$ becomes large. Hence, only good leverage points can have a large effect on the test statistics. Since $\gamma_j < 2b_j$ and the constants $\eta_j$ and $\gamma_j$ are always positive, the impact of contamination is larger for $\Lambda_{\text{MM}}$ than for $\Lambda_{\text{S}}$. The increased efficiency of MM-estimators thus implies some loss in robustness.


\subsection{Robust Test for Diagonality of the Covariance Matrix}

\noindent The following theorem shows that under the null hypothesis in~\eqref{HypothesisTestSigma} the asymptotic distribution of the robust test statistics $\text{LM}_{\text{S}}$ and $\text{LM}_{\text{MM}}$ is proportional to a chi-squared distribution.

\begin{theorem}
\label{TheoremAsymptoticDistributionLM}
Let $Z=(\tilde{X}^\top,Y^\top)^\top$ have model distribution $H_{\beta,\Sigma}$. Assume that $\Sigma$ is a diagonal matrix. Then, it holds that
\[ \text{LM}_{\text{S}} \quad \overset{d}{\longrightarrow} \quad \frac{m}{(m+2) \gamma_0^2} {\rm E}_F [\psi_0^2(\norm{\mathcal{E}}_{\Sigma}) \norm{\mathcal{E}}_{\Sigma}^2] \hspace{1mm} \chi_{m(m-1)/2}^2, \]
and
\[ \text{LM}_{\text{MM}} \quad \overset{d}{\longrightarrow} \quad \frac{m}{(m+2) \gamma_1^2} {\rm E}_F [\psi_1^2(\norm{\mathcal{E}}_{\Sigma}) \norm{\mathcal{E}}_{\Sigma}^2] \hspace{1mm} \chi_{m(m-1)/2}^2. \]
\end{theorem}

To investigate the robustness of the resulting tests, we again derive the second-order influence function of the test statistics under the null hypothesis.

\begin{theorem}
\label{TheoremInfluenceFunctionsLM}
If $Z=(\tilde{X}^\top,Y^\top)^\top$ has model distribution $H_{\beta,\Sigma}$ and if $\Sigma$ is a diagonal matrix, then the second-order influence function of $\text{LM}_{\text{S}}$ and $\text{LM}_{\text{MM}}$ are given by
\[ \text{IF2}(z;\text{LM}_{\text{S}},H) = \frac{2 m^2}{\gamma_0^2} w_0^2(\norm{e}_{\Sigma}) \sum_{j<k} \frac{e_j^2 e_k^2}{\sigma_{jj} \sigma_{kk}}, \]
and
\[ \text{IF2}(z;\text{LM}_{\text{MM}},H) = \frac{2 m^2}{\gamma_1^2} w_1^2(\norm{e}_{\Sigma}) \sum_{j<k} \frac{e_j^2 e_k^2}{\sigma_{jj} \sigma_{kk}}. \]
\end{theorem}

This theorem shows that leverage points do not influence the test statistic and that large response outliers have zero influence as well. The boundedness of the second-order influence functions ensures the stability of the asymptotic level and power of these diagonality tests~\citep{Heritier1994}.


\section{Finite-Sample Performance of diagonality test (continued)}

In Section~\ref{Finite-Sample Performance} of the manuscript we have investigated the power of the diagonality test for the situation where $\Sigma$ is an equicorrelation matrix with correlation $\tau$ ranging from 0 to 0.5 with step length 0.1. While for this setting the deviation from diagonality was present in all covariance elements, we now consider a situation that is less diverging from diagonality. In particular, we consider the same simulation setting as in the final paragraph of Section~\ref{Finite-Sample Performance} but now set $\Sigma$ equal to
\[ \begin{bmatrix}
1 & \tau & 0 \\
\tau & 1 & 0 \\
0 & 0 & 1 \\
\end{bmatrix}, \]
where $\tau$ takes values from 0 to 0.5 with step length 0.1. For data simulated under this alternative hypothesis, the resulting power curves are shown in Figure~\ref{paper_simulationFRBsurLM_diagtestH1_extra}.
\begin{figure}[ht!]
\centering
\includegraphics[width=\textwidth, trim= 0mm 152.4mm 0mm 0mm, clip]{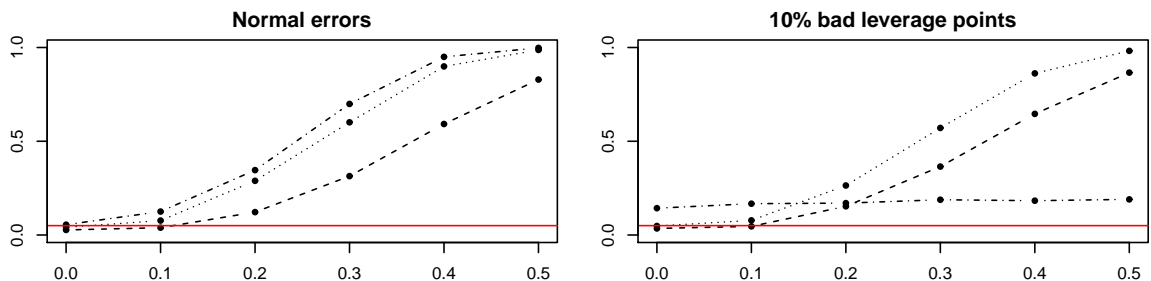}
\caption{Power curves of the hypothesis test in~\eqref{HypothesisTestSigma} based on the test statistics $\text{LM}_{\text{S}}$ (dashed), $\text{LM}_{\text{MM}}$ (dotted) and $\text{LM}_{\text{MLE}}$ (dash-dotted). The solid (red) line represents the rejection level of 5\%.}
\label{paper_simulationFRBsurLM_diagtestH1_extra}
\end{figure}
The left plot corresponds to the situation with normal errors without outliers. The right panel shows the power curves in case 10\% contamination is added to the data as in Section~\ref{Finite-Sample Performance} of the manuscript.  Compared to the equicorrelation setting considered in the manuscript, all power curves increase at a slower pace because we are now considering a difficult case where only one correlation is responsible for the deviation from diagonality. 
However, when comparing the classical and robust Breusch-Pagan tests, the same conclusions can be drawn as in the manuscript. In absence of contamination the test based on MM-estimators performs almost as well as the classical test. Moreover, in contrast to the classical test its performance is not much affected in the presence of bad leverage points. 


\section{Robust Confidence Intervals}

\noindent The results in Theorem~\ref{TheoremAsymptoticVariance} can be used to construct confidence intervals for the parameters in the SUR model based on their MM-estimates. For example, a $100(1-\alpha)\%$ confidence interval for a regression parameter $\beta_{kl}$ can be obtained as
\begin{equation}
\label{AsymptoticConfidenceInterval}
[\hat{\beta}_{kl} - z_{1-\alpha/2} \sqrt{\text{ASV}(\hat{\beta}_{kl},\hat{H}_n)/n}, \hat{\beta}_{kl} + z_{1-\alpha/2} \sqrt{\text{ASV}(\hat{\beta}_{kl},\hat{H}_n)/n}],
\end{equation}
with $z_{\alpha}$ the $\alpha$ quantile of the standard normal distribution and $\text{ASV}(\hat{\beta}_{kl},\hat{H}_n)$ an estimate of the asymptotic variance in~\eqref{AsymptoticVariancebeta} based on the empirical distribution corresponding to the data.

Alternatively, regular bootstrap confidence intervals are constructed as follows. Let $(\hat{\beta}_{kl}^*)_1,\ldots,(\hat{\beta}_{kl}^*)_N$ be a set of $N$ parameter estimates based on bootstrap samples. Then a $100(1-\alpha)\%$ percentile confidence interval for $\beta_{kl}$ is obtained as
\begin{equation*}
[(\hat{\beta}_{kl}^*)_{((N+1)\alpha_L)}, (\hat{\beta}_{kl}^*)_{((N+1)\alpha_R)}],
\end{equation*}
where $(\hat{\beta}_{kl}^*)_{(.)}$ denotes the order statistics corresponding to the bootstrap estimates. Basic percentile (BP) confidence intervals select $\alpha_L=\alpha/2$ and $\alpha_R=1-\alpha/2$. To improve the accuracy of the confidence intervals, the bias-corrected and accelerated (BCa) method~\citep{Efron1987} can be used to determine the confidence levels $\alpha_L$ and $\alpha_R$. See, e.g., \citet{Davison1997} for more details on percentile methods.

As explained in Section~\ref{Fast and Robust Bootstrap} of the manuscript, confidence intervals based on standard bootstrap are not attractive because they are not robust. Therefore, we propose to construct bootstrap confidence intervals based on the FRB estimates. For example, a $100(1-\alpha)\%$ FRB percentile confidence interval for $\beta_{kl}$ is computed as
\begin{equation}
\label{FRBPercentileConfidenceInterval}
[(\hat{\beta}_{kl}^{R*})_{((N+1)\alpha_L)}, (\hat{\beta}_{kl}^{R*})_{((N+1)\alpha_R)}],
\end{equation}
where $(\hat{\beta}_{kl}^{R*})_1,\ldots,(\hat{\beta}_{kl}^{R*})_N$ is a set of $N$ FRB bootstrap replicates.


\section{Finite-Sample Performance for Confidence Intervals}

\noindent The performance of confidence intervals obtained by FRB based on robust S and MM-estimators for the SUR model is investigated by simulation. We focus on intervals for the regression coefficients $\beta$ with 95\% confidence level. The performance is measured by their coverage and their average length.

We consider the same simulation setting as in Section~\ref{Finite-Sample Performance} of the manuscript. Robust S-estimators and MM-estimators are computed with maximal breakdown point of 50\% and the MM-estimator has 90\% efficiency. $N=1000$ bootstrap samples are generated for the FRB. Three different confidence intervals are calculated for the slopes in the model: asymptotic confidence intervals (AS) given by~\eqref{AsymptoticConfidenceInterval}, and BP and BCa confidence intervals according to~\eqref{FRBPercentileConfidenceInterval} based on FRB. For each simulation setting the coverage is estimated by the fraction of the confidence intervals that contains the true value of the parameter. The reported coverage and average lengths of the confidence intervals are the average results for all slopes in the model. In Figure~\ref{paper_simulationFRBsurMM_cover} the coverage is depicted as a function of sample size, while the average interval lengths are given in Tables~\ref{simulation_FRBsurMM_meanlength_normal} and~\ref{simulation_FRBsurMM_meanlength_cont10} for data with normal errors, containing 0\% or 10\% of contamination, respectively.
\begin{figure}[!ht]
\centering
\includegraphics[width=\textwidth]{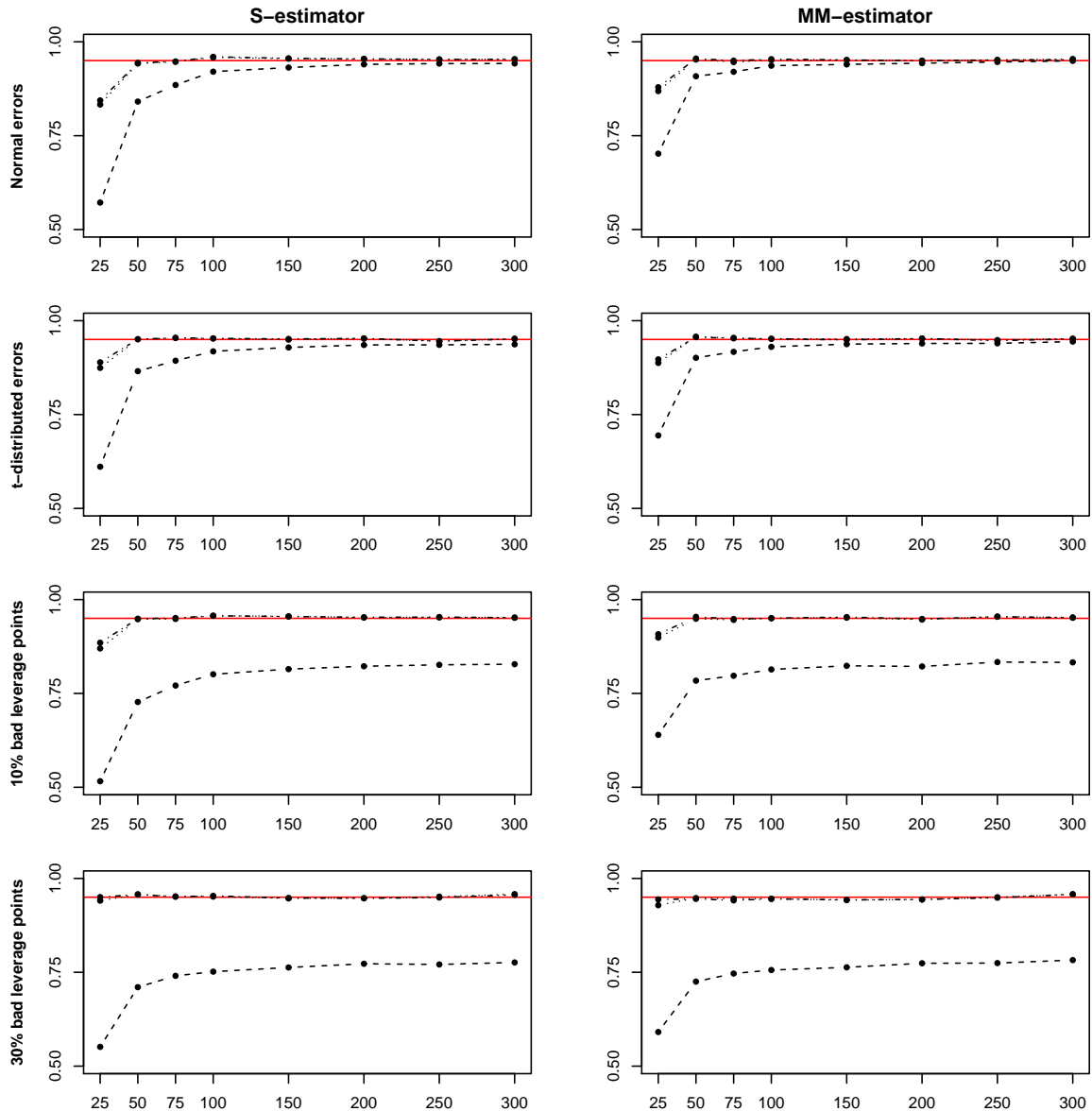}
\caption{Coverage results of 95\% confidence intervals obtained by the AS (dashed), BP (dash-dotted) and BCa (dotted) methods. The solid (red) line represents the nominal level of 95\%.}
\label{paper_simulationFRBsurMM_cover}
\end{figure}
\begin{table}[ht!]
\renewcommand{\arraystretch}{1.25}
\begin{center}
\begin{tabular}{l l c c c c c c c c}
\hline
\textbf{Estimator} & \textbf{Type} & \multicolumn{8}{c}{\textbf{Sample size $n$}} \\
& & 25 & 50 & 75 & 100 & 150 & 200 & 250 & 300 \\
\hline
  & AS  & 0.467 & 0.470 & 0.406 & 0.359 & 0.298 & 0.260 & 0.234 & 0.214 \\
S & BP  & 1.315 & 0.786 & 0.538 & 0.437 & 0.336 & 0.283 & 0.250 & 0.226 \\
  & BCa & 1.319 & 0.787 & 0.539 & 0.438 & 0.338 & 0.284 & 0.251 & 0.227 \\
\hline
   & AS  & 0.533 & 0.458 & 0.380 & 0.331 & 0.272 & 0.236 & 0.212 & 0.194 \\
MM & BP  & 1.064 & 0.572 & 0.431 & 0.362 & 0.288 & 0.245 & 0.218 & 0.198 \\
   & BCa & 1.126 & 0.573 & 0.433 & 0.364 & 0.289 & 0.246 & 0.219 & 0.199 \\
\hline
\end{tabular}
\end{center}
\caption{Average length of 95\% confidence intervals obtained with the AS, BP and BCa methods for normal errors and without contamination.}
\label{simulation_FRBsurMM_meanlength_normal}
\end{table}
\begin{table}[ht!]
\renewcommand{\arraystretch}{1.25}
\begin{center}
\begin{tabular}{l l c c c c c c c c}
\hline
\textbf{Estimator} & \textbf{Type} & \multicolumn{8}{c}{\textbf{Sample size $n$}} \\
& & 25 & 50 & 75 & 100 & 150 & 200 & 250 & 300 \\
\hline
  & AS  & 0.387 & 0.349 & 0.298 & 0.261 & 0.216 & 0.188 & 0.169 & 0.155 \\
S & BP  & 1.506 & 0.751 & 0.526 & 0.430 & 0.334 & 0.282 & 0.250 & 0.226 \\
  & BCa & 1.543 & 0.754 & 0.528 & 0.432 & 0.336 & 0.283 & 0.251 & 0.227 \\
\hline
   & AS  & 0.426 & 0.340 & 0.283 & 0.246 & 0.202 & 0.175 & 0.157 & 0.144 \\
MM & BP  & 1.084 & 0.587 & 0.444 & 0.374 & 0.297 & 0.254 & 0.226 & 0.205 \\
   & BCa & 1.108 & 0.590 & 0.446 & 0.376 & 0.298 & 0.255 & 0.227 & 0.206 \\
\hline
\end{tabular}
\end{center}
\caption{Average length of 95\% confidence intervals obtained with the AS, BP and BCa methods for normal errors and 10\% bad leverage points.}
\label{simulation_FRBsurMM_meanlength_cont10}
\end{table}
From Figure~\ref{paper_simulationFRBsurMM_cover} we can see that the coverage for S and MM-estimators is very similar for all settings. These results clearly show that the FRB confidence intervals reach the nominal 95\% coverage level much sooner (for $n \geq 50$ already) then the asymptotic confidence intervals, which for $n=300$ still haven't completely reached the nominal level. Moreover, there is almost no difference between the two types of FRB percentile confidence intervals. Hence, the more complex BCa intervals do not seem to offer any gain over the more simple basic percentile intervals in this case. For the situations with normal and t-distributed errors, the coverage converges to the nominal level for all three methods. On the other hand, for data contaminated with 10\% or 30\% bad leverage points in each block, the asymptotic confidence intervals fail to get close to 95\% coverage while the FRB confidence intervals still reach the nominal level quickly. This clearly shows the robustness of the FRB based confidence intervals.

Using MM-estimators does not yield confidence intervals with better coverage compared to S-estimators. However, as can be seen from Table~\ref{simulation_FRBsurMM_meanlength_normal} and~\ref{simulation_FRBsurMM_meanlength_cont10}, confidence intervals based on MM-estimators are generally shorter than those based on S-estimators. The increased efficiency of the MM-estimators thus leads to more informative confidence intervals. These tables also show that the asymptotic confidence intervals are much shorter than the FRB intervals in all cases. However, these intervals are too short, resulting in (severe) under-coverage as seen in Figure~\ref{paper_simulationFRBsurMM_cover}. Note that in terms of average length there is again little difference between the BCa and BP confidence intervals. Finally, by comparing the two tables it can be seen that 10\% of bad leverage points does not affect the average length of the FRB confidence intervals much in this setting.

Similarly as for the regression coefficients, confidence intervals for the elements of the scatter matrix $\Sigma$ or shape matrix $\Gamma$ and scale $\sigma$ can be constructed. For the shape matrix, the behavior of the confidence intervals is the same as for the regression coefficients. For the scale and the elements of the scatter matrix the performance is generally worse in presence of contamination. The reason is that the scale S-estimator is not redescending and thus contamination has a persistent effect on the scale estimate which also affects the confidence intervals.

In summary, we can conclude that asymptotic confidence intervals only yield reliable results for clean data with large sample size while FRB confidence intervals remain reliable for contaminated data and smaller sample sizes. Moreover, MM-estimators yield more informative inference than S-estimators.


\section{Example: Grunfeld Data (Continued)}
\label{Example: Grunfeld Data (Continued}

\noindent
As in Section~\ref{Example: Grunfeld Data} of the manuscript we use the Grunfeld data and consider a SUR model with three blocks corresponding to the U.S.\ corporations General Electric (GE), Westinghouse (W), and Diamond Match (DM). The SUR model is given in~\eqref{Grunfeld_model}. As before, the MM-estimates are calculated with 50\% breakdown point and a normal efficiency of 90\%. The robust coefficient estimates (and their bootstrap standard errors) are presented in Table~\ref{Grunfeld_beta_se} of the manuscript.

We consider the construction of confidence intervals corresponding to the robust MM-estimators. Confidence intervals are computed based on asymptotic results and the fast and robust bootstrap. For the FRB $N=999$ bootstrap samples are generated using case resampling.

We now compare inference results for regression coefficients in the SUR model. As an example, we first focus on $\beta_{22}$, the slope for predictor Capital in the regression equation for Westinghouse. A histogram of the FRB replications of $\hat{\beta}_{22}$ is presented in Figure~\ref{paper_Grunfeld_bootstraphist_MMbeta22}.
\begin{figure}[ht!]
\centering
\includegraphics[width=\textwidth]{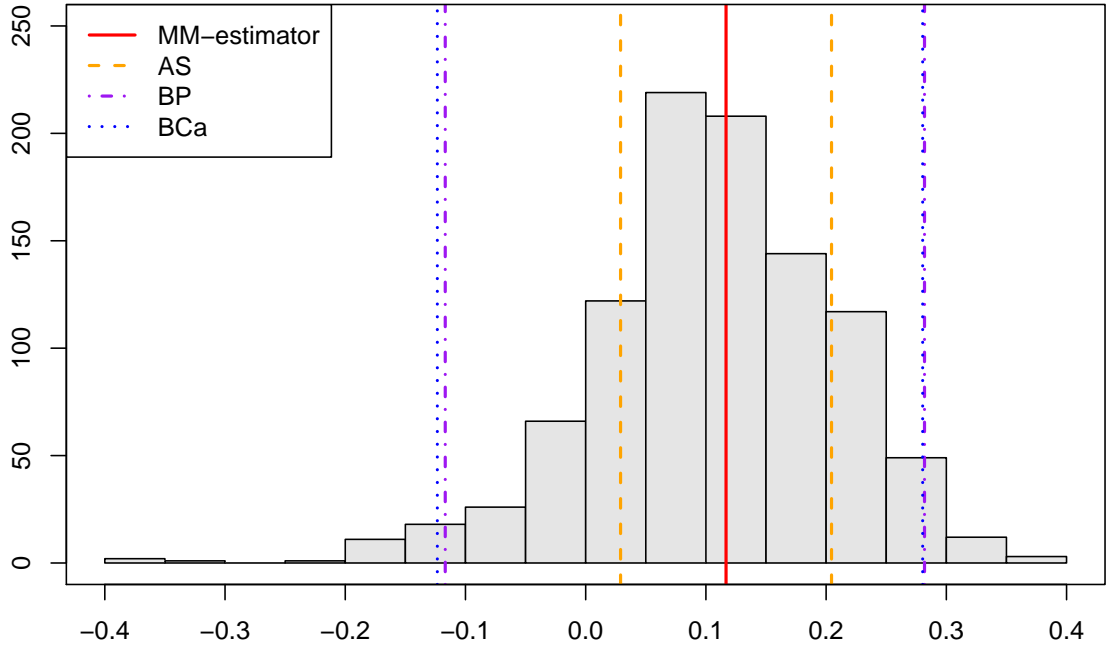}
\caption{Histogram of FRB replications of $\hat{\beta}_{22}$ in the SUR model for the Grunfeld data. The solid (red) line corresponds to the MM-estimate $\hat{\beta}_{22}$. Three $95\%$ confidence intervals based on the MM-estimate are shown. The dashed (orange) lines represent the boundaries of the asymptotic confidence interval. The dash-dotted (purple) and dotted (blue) lines show the bound of the BP and BCa confidence intervals, respectively.}
\label{paper_Grunfeld_bootstraphist_MMbeta22}
\end{figure}
The solid (red) vertical line corresponds to the MM-estimate of this coefficient as reported in Table~\ref{Grunfeld_beta_se}. The dashed (orange) lines represent the asymptotic $95\%$ confidence interval based on the MM-estimate as given by~\eqref{AsymptoticConfidenceInterval}. The dash-dotted (purple) and dotted (blue) lines represent the BP and BCa confidence intervals as given by~\eqref{FRBPercentileConfidenceInterval}, respectively. It can immediately be seen that the asymptotic confidence interval which relies on the assumption that the distribution of $\hat{\beta}_{22}$ is a normal distribution is much narrower than the other two. However, from the histogram of the FRB replications it is clear that the bootstrap distribution is skewed, which indicates that the normality assumption is not realistic. The two bootstrap confidence intervals do not rely on the normality assumption and they can also better resist the effect of outliers, which makes them more reliable in this case. When checking significance of this regression coefficient, the bootstrap confidence intervals yield a different conclusion than the asymptotic confidence interval. Indeed, both bootstrap percentile confidence intervals contain zero, implying that the coefficient is non-significant, but based on the asymptotic confidence interval the coefficient would be considered significant. However, as seen in the simulations, the asymptotic confidence interval is most likely too small, leading to under-coverage and too optimistic conclusions.

The three confidence intervals for each of the regression coefficients are reported in Table~\ref{Grunfeld_beta_conf}.
\begin{table}[ht!]
\renewcommand{\arraystretch}{1.25}
\begin{center}
\begin{tabular}{l l r@{.}l  r@{.}l r@{.}l  r@{.}l r@{.}l r@{.}l}
\hline
\textbf{Corporation} & \textbf{Coefficient} & \multicolumn{4}{c}{\textbf{AS}} & \multicolumn{4}{c}{\textbf{BP}} & \multicolumn{4}{c}{\textbf{BCa}} \\
 & & \multicolumn{2}{r}{lower} & \multicolumn{2}{r}{upper} & \multicolumn{2}{r}{lower} & \multicolumn{2}{r}{upper} & \multicolumn{2}{r}{lower} & \multicolumn{2}{r}{upper} \\
\hline
\multirow{3}{*}{GE} & $\beta_{01}$ & -68&941 & 7&619 & -84&541 & 21&659 & -89&872 &  17&221 \\
                    & $\beta_{11}$ &   0&015 & 0&051 &   0&010 &  0&069 &   0&012 &   0&071 \\
                    & $\beta_{21}$ &   0&117 & 0&187 &   0&083 &  0&184 &   0&096 &   0&186 \\
\hline
\multirow{3}{*}{W}  & $\beta_{02}$ & -19&106 & 6&465 & -25&379 & 18&015 & -34&574 &  10&736 \\
                    & $\beta_{12}$ &   0&035 & 0&082 &   0&017 &  0&105 &   0&027 &   0&121 \\
                    & $\beta_{22}$ &   0&029 & 0&204 &  -0&117 &  0&282 &  -0&123 &   0&280 \\
\hline
\multirow{3}{*}{DM} & $\beta_{03}$ &  -2&253 & 0&543 &  -2&187 &  0&110 &  -2&167 &   0&170 \\
                    & $\beta_{13}$ &  -0&016 & 0&020 &  -0&011 &  0&021 &  -0&013 &   0&019 \\
                    & $\beta_{23}$ &   0&554 & 0&674 &   0&411 &  0&773 &   0&390 &   0&757 \\
\hline
\end{tabular}
\end{center}
\caption{Three $95\%$ confidence intervals (AS, BP and BCa) for the regression coefficients in the SUR model for the Grunfeld data, based on their MM-estimates.}
\label{Grunfeld_beta_conf}
\end{table}
As already seen in the simulations, both percentile confidence intervals are very similar, while the asymptotic confidence intervals are generally much shorter. This illustrates again that asymptotic confidence intervals may lead to unreliable conclusions.


\section*{Appendix}


\noindent \textbf{Partial derivatives of $g$.} In order to apply the fast and robust bootstrap procedure the partial derivatives $\nabla g$ need to be computed. The Jacobian of $g = (g_1^\top, g_2^\top, g_3^\top, g_4^\top)^\top$ given in equation~\eqref{FRBgfunctionMMestimator} has the following form
\[ \nabla g =
\begin{bmatrix}
\diff(g_1,\hat{\beta}) & \diff(g_1,\hat{\Gamma}) & \diff(g_1,\tilde{\Sigma}) & 0                        \\[3mm]
\diff(g_2,\hat{\beta}) & \diff(g_2,\hat{\Gamma}) & \diff(g_2,\tilde{\Sigma}) & 0                        \\[3mm]
0                      & 0                       & \diff(g_3,\tilde{\Sigma}) & \diff(g_3,\tilde{\beta}) \\[3mm]
0                      & 0                       & \diff(g_4,\tilde{\Sigma}) & \diff(g_4,\tilde{\beta})
\end{bmatrix}. \]
Note that the two upper rows in this gradient correspond to the estimating equations of the MM-estimator, while the two bottom rows correspond to those of the S-estimator. The expressions for the S-estimator are omitted because these are similar to the derivatives for the MM-estimator. Consider the matrices
\[ A_i = {\rm bdiag}(a_i,\ldots,a_i) \quad i=1,\ldots,n \]
where the vector $a_i$ is repeated $m$ times. The vector $a_i$ has length $n$ and is defined as $a_i = (0, \ldots, 0, 1, 0, \ldots, 0)^\top$ with the 1 at the $i$th entry of the vector. Write $\tilde{y}_i = A_i^\top y$ and $x_i = A_i^\top X$, that is, $\tilde{y}_i$ and $x_i$ contain the information of the $i$th observation across all blocks. Introduce the following notation as well:
\begin{alignat*}{2}
U \quad &= \quad X^\top ( \hat{\Sigma}^{-1} \kron D) X \quad && (p \times p) \\
W \quad &= \quad X^\top ( \hat{\Sigma}^{-1} \kron D) y \quad && (p \times 1) \\
T \quad &= \quad (W \kron I_p)^\top (U^{-1} \kron U^{-1}) \quad && (p \times p^2) \\
V \quad &= \quad (Y - \tilde{X} \hat{\mathcal{B}})^\top D (Y - \tilde{X} \hat{\mathcal{B}}) \quad && (m \times m) \\
S \quad &= \quad \abs{V}^{-1/m} \left( I_{m^2} - \frac{1}{m} {\rm vec} (V) {\rm vec} (V^{-1})^\top \right) \quad && (m^2 \times m^2)
\end{alignat*}
Straightforward derivations then lead to the following expressions:
\begin{align*}
\diff(g_1,\hat{\beta}) &= T \Bigg( \sum_{i=1}^n \frac{w_1'(d_i)}{d_i \tilde{\sigma}^2} {\rm vec} (x_i^\top \hat{\Gamma}^{-1} x_i) (x_i^\top \hat{\Gamma}^{-1} e_i)^\top \Bigg) - U^{-1} \Bigg( \sum_{i=1}^n \frac{w_1'(d_i)}{d_i \tilde{\sigma}^2} (x_i^\top \hat{\Gamma}^{-1} \tilde{y}_i) (x_i^\top \hat{\Gamma}^{-1} e_i)^\top \Bigg), \\
\diff(g_1,\hat{\Gamma}) &= T \Bigg( \sum_{i=1}^n w_1(d_i) (x_i^\top \kron x_i^\top) (\hat{\Gamma}^{-1} \kron \hat{\Gamma}^{-1}) + \frac{w_1'(d_i)}{2 d_i \tilde{\sigma}^2} {\rm vec} (x_i^\top \hat{\Gamma}^{-1} x_i) {\rm vec} (\hat{\Gamma}^{-1} e_i e_i^\top \hat{\Gamma}^{-1})^\top \Bigg) \\
& \hspace{5mm} - U^{-1} \Bigg( \sum_{i=1}^n w_1(d_i) (\tilde{y}_i^\top \kron x_i^\top) (\hat{\Gamma}^{-1} \kron \hat{\Gamma}^{-1}) + \frac{w_1'(d_i)}{2 d_i \tilde{\sigma}^2} (x_i^\top \hat{\Gamma}^{-1} \tilde{y}_i) {\rm vec} (\hat{\Gamma}^{-1} e_i e_i^\top \hat{\Gamma}^{-1})^\top \Bigg), \\
\diff(g_1,\tilde{\Sigma}) &= T \left( \sum_{i=1}^n \frac{w_1'(d_i) d_i}{2 m \tilde{\sigma}^2} {\rm vec} (x_i^\top \hat{\Gamma}^{-1} x_i) {\rm vec} (\tilde{\Gamma}^{-1})^\top \right) - U^{-1} \left( \sum_{i=1}^n \frac{w_1 (d_i) d_i}{2 m \tilde{\sigma}^2} (x_i^\top \hat{\Gamma}^{-1} \tilde{y}_i) {\rm vec} (\tilde{\Gamma}^{-1})^\top \right), \\
\diff(g_2,\hat{\beta}) &= - S \Bigg( \sum_{i=1}^n w_1(d_i) \big( x_i \kron \tilde{y}_i + \tilde{y}_i \kron x_i - (x_i \kron x_i) (I_p \kron \hat{\beta} + \hat{\beta} \kron I_p) \big) + \frac{w_1'(d_i)}{d_i \tilde{\sigma}^2} {\rm vec} (e_i e_i^\top) (x_i^\top \hat{\Gamma}^{-1} e_i)^\top \Bigg), \\
\diff(g_2,\hat{\Gamma}) &= - S \left( \sum_{i=1}^n \frac{w_1'(d_i)}{2 d_i \tilde{\sigma}^2} {\rm vec} (e_i e_i^\top) {\rm vec} (\hat{\Gamma}^{-1} e_i e_i^\top \hat{\Gamma}^{-1})^\top \right), \\
\diff(g_2,\tilde{\Sigma}) &= - S \left( \sum_{i=1}^n \frac{w_1'(d_i) d_i}{2 m \tilde{\sigma}^2} {\rm vec} (e_i e_i^\top) {\rm vec} (\tilde{\Gamma}^{-1})^\top \right).
\end{align*}


\noindent \textbf{Consistency condition of $\Lambda_{\text{S}}$ and $\Lambda_{\text{MM}}$.} Consider the $h$ function of $\Lambda_{\text{S}}$ defined through equations~\eqref{RobustLikelihoodRatioS_consistent} and~\eqref{MultivariateMscale} (a similar derivation holds for $\Lambda_{\text{MM}}$). In order for the partial derivatives of $h$ to vanish it is sufficient to show that the partial derivatives of $\tilde{s}(b,G)$ converge to zero for $(\tilde{\beta},\tilde{\Gamma})$. Differentiating~\eqref{MultivariateMscale2} with respect to $b$ leads to
\[ \frac{1}{n} \sum_{i=1}^n \frac{\psi_0 (d_i(b,G))}{\tilde{s}^2(b,G)} \left( \diff(,b) \left( \sqrt{e_i^\top(b) \phi(G^{-1}) e_i(b)} \right) \tilde{s}(b,G) - d_i \diff({\tilde{s}(b,G)},b) \right) = 0. \]
Rearranging terms and evaluating the inner derivative we obtain
\[ \diff({\tilde{s}(b,G)},b) = - \left( \sum_{i=1}^n \frac{w_0(d_i(b,G))}{\tilde{s}^2(b,G)} X_i^\top \phi(G^{-1}) e_i(b) \right) \left( \sum_{i=1}^n \frac{\psi_0 (d_i(b,G)) d_i(b,G)}{\tilde{s}(b,G)} \right)^{-1}, \]
which is exactly zero for $(b,G) = (\tilde{\beta},\tilde{\Gamma})$ due to the estimating equations of $\tilde{\beta}$.

Differentiating~\eqref{MultivariateMscale2} with respect to $G$ leads to
\[ \frac{1}{n} \sum_{i=1}^n \frac{\psi_0 (d_i(b,G)) d_i(b,G)}{\tilde{s}(b,G)} \diff({\tilde{s}(b,G)},G) = \frac{1}{n} \sum_{i=1}^n \frac{\psi_0 (d_i(b,G))}{2 \tilde{s}^2(b,G) d_i(b,G)} \diff(,G) \left( \sqrt{e_i^\top(b) \phi(G^{-1}) e_i(b)} \right). \]
The right hand-side of this equality can be simplified to
\[ - \abs{G^{-1}}^{-1/m} G^{-1} \left( \sum_{i=1}^n \frac{\psi_0 (d_i(b,G))}{2 \tilde{s}^2(b,G) d_i(b,G)} e_i(b) e_i^\top(b) \right) G^{-1} + \frac{1}{m} \abs{G^{-1}}^{-1/m} G^{-1} \sum_{i=1}^n \frac{\psi_0 (d_i(b,G)) d_i(b,G)}{2}. \]
By evaluating the previous line at $(b,G) = (\tilde{\beta},\tilde{\Gamma})$ and using the estimating equations for $\tilde{\Sigma}$, it can be shown that the right hand-side reduces to zero. \vspace{0.5cm}


\begin{proof}[Proof of Theorem~\ref{TheoremInfluenceFunctions}]
Let $Z=(\tilde{X}^\top,Y^\top)^\top$ have model distribution $H=(K,F)$ with $K$ the distribution of $\tilde{X}$ and $F:=F_\Sigma$ the elliptically symmetric distribution of $Y$. Let us denote $H_{\epsilon} = H_{\epsilon,\Delta_{z}}$ to simplify the notation. The influence function of $\hat{\beta}(H)$ is obtained by differentiating the estimating equations for $\hat{\beta}(H_{\epsilon})$ w.r.t.\ $\epsilon$ and evaluating the result at $\epsilon = 0$. The derivation of these equations is similar as in the finite-sample case. For a general distribution function $H$ of $Z=(\tilde{X}^\top,Y^\top)^\top$, the estimating equations of the MM-functionals $\hat{\beta}(H)$ and $\hat{\Sigma}(H)$ are given by
\begin{gather*}
\int w_1(d(H)) x^\top \hat{\Sigma}^{-1}(H) e(H) \hspace{0.5mm} dH(z) = 0 \\
\hat{\Sigma}(H) \int \psi_1(d(H)) d(H) \hspace{0.5mm} dH(z) = m \int w_1(d(H)) e(H) e(H)^\top \hspace{0.5mm} dH(z)
\end{gather*}
where $d^2(H) = e(H)^\top \hat{\Sigma}^{-1}(H) e(H)$ and $e(H) = y - x \hat{\beta}(H)$. We thus have
\[ \diff(,\epsilon) \left[ \int w_1(d(H_{\epsilon})) x^\top \hat{\Sigma}^{-1}(H_{\epsilon}) e(H_{\epsilon}) \hspace{0.5mm} dH_{\epsilon}(z) \right] \Bigg{|}_{\epsilon = 0} = 0, \]
which can be rewritten as
\[ \diff(,\epsilon) \bigg[ (1-\epsilon) \int w_1(d(H_{\epsilon})) x^\top \hat{\Sigma}^{-1}(H_{\epsilon}) e(H_{\epsilon}) \hspace{0.5mm} dH(z) + \epsilon \int w_1(d(H_{\epsilon})) x^\top \hat{\Sigma}^{-1}(H_{\epsilon}) e(H_{\epsilon}) \hspace{0.5mm} d\Delta_{z}(z) \bigg] \Bigg{|}_{\epsilon = 0} = 0. \]
Applying the chain rule and using the estimating equation at $H$ yields
\[ \diff(,\epsilon) \left[ \int w_1(d(H_{\epsilon})) x^\top \hat{\Sigma}^{-1}(H_{\epsilon}) e(H_{\epsilon}) \hspace{0.5mm} dH(z) \right] \Bigg{|}_{\epsilon = 0} + \int w_1(d(H)) x^\top \hat{\Sigma}^{-1}(H) e(H) \hspace{0.5mm} d\Delta_{z} (z) = 0. \]
The second term simplifies to $w_1(\norm{y}_{\Sigma}) x^\top \Sigma^{-1} y$. Differentiation of the first term and symmetry of $F$ yields
\[\int \diff(,\epsilon) \left( w_1(d(H_{\epsilon})) \right) \Big{|}_{\epsilon = 0} x^\top \Sigma^{-1} y \hspace{0.5mm} dH(z) + \int w_1(\norm{y}_{\Sigma}) x^\top \Sigma^{-1} \diff(,\epsilon) \left( e(H_{\epsilon}) \right) \Big{|}_{\epsilon = 0} \hspace{0.5mm} dH(z). \]
Computing the inner derivatives and simplifying the result leads to
\begin{equation}
- \int \frac{w_1'(\norm{y}_{\Sigma})}{\norm{y}_{\Sigma}} y^\top \Sigma^{-1} x \text{IF}(z;\hat{\beta},H) x^\top \Sigma^{-1} y \hspace{0.5mm} dH(z) - \int w_1(\norm{y}_{\Sigma}) x^\top \Sigma^{-1} x \hspace{0.5mm} dH(z) \hspace{0.5mm} \text{IF}(z;\hat{\beta},H).
\label{A.1}
\end{equation}
Splitting the first integral into a $\tilde{x}$ and a $y$ component yields
\[ \int x^\top \Sigma^{-1/2} \left( \int \frac{w_1'(\norm{y}_{\Sigma})}{\norm{y}_{\Sigma}} \Sigma^{-1/2} y y^\top \Sigma^{-1/2} \hspace{0.5mm} dF(y) \right) \Sigma^{-1/2} x \hspace{0.5mm} dK(\tilde{x}) \hspace{0.5mm} \text{IF}(z;\hat{\beta},H). \]
Using symmetry and results in~\citep{Lopuhaa1999} this can be rewritten as
\[ \int x^\top \Sigma^{-1/2} \left( \frac{1}{m} \int w_1'(\norm{y}_{\Sigma}) \norm{y}_{\Sigma} \hspace{0.5mm} dF(y) \hspace{0.5mm} I_m \right) \Sigma^{-1/2} x \hspace{0.5mm} dK(\tilde{x}) \hspace{0.5mm} \text{IF}(z;\hat{\beta},H). \]
Combining both integrals in~\eqref{A.1} now yields
\[ - {\rm E}_F \left[ \frac{1}{m} w_1'(\norm{Y}_{\Sigma}) \norm{Y}_{\Sigma} + w_1(\norm{Y}_{\Sigma}) \right] {\rm E}_K [X^\top \Sigma^{-1} X] \hspace{0.5mm} \text{IF}(z;\hat{\beta},H) \\
+ w_1(\norm{y}_{\Sigma}) x^\top \Sigma^{-1} y = 0. \]
Rearranging terms leads to the result in~\eqref{InfluenceFunctionbeta}.
\end{proof}


\begin{proof}[\textbf{Proof of Theorem~\ref{TheoremAsymptoticVariance}}]
Consider $Z=(\tilde{X}^\top,Y^\top)^\top$ with model distribution $H:=H_{0,\Sigma}$. The asymptotic variance of $\hat{\beta}$ is given by
\[ \text{ASV} (\hat{\beta},H) = \int \text{IF}(z;\hat{\beta},H) \text{IF}^\top(z;\hat{\beta},H) \hspace{0.5mm} dH(z). \]
Using the expression for the influence function in~\eqref{InfluenceFunctionbeta}, we obtain
\[ \int \frac{1}{\eta_1^2} w_1^2(\norm{y}_{\Sigma}) {\rm E}_K [X^\top \Sigma^{-1} X]^{-1} x^\top \Sigma^{-1} y y^\top \Sigma^{-1} x {\rm E}_K [X^\top \Sigma^{-1} X]^{-1} \hspace{0.5mm} dH(z). \]
Splitting the remaining integral yields
\[ \int x^\top \Sigma^{-1/2} \left( \int w_1^2(\norm{y}_{\Sigma}) \Sigma^{-1/2} y y^\top \Sigma^{-1/2} \hspace{0.5mm} dF(y) \right) \Sigma^{-1/2} x \hspace{0.5mm} dK(\tilde{x}), \]
which by symmetry can be rewritten as
\[ \frac{\alpha_1}{m} {\rm E}_K [X^\top \Sigma^{-1} X]. \]
Combining the results yields~\eqref{AsymptoticVariancebeta}. For a general distribution $H_{\beta,\Sigma}$ this result is obtained by using the affine equivariance property.
\end{proof}


To proof Theorem~\ref{TheoremAsymptoticDistributionLambda}, we need the following lemma.

\begin{lemma}
\label{Lemmafirstorderapprox}
Under $H_0: R \beta = q$ and the conditions of Theorem~\ref{TheoremAsymptoticDistributionLambda}, it holds that
\begin{equation}
\label{firstorderapprox}
\sqrt{n} (\tilde{\beta} - \tilde{\beta}_r) = \sqrt{n} {\rm E}_K [X^\top \Sigma^{-1} X]^{-1} R^\top \left( R \hspace{0.5mm} {\rm E}_K [X^\top \Sigma^{-1} X]^{-1} R^\top \right)^{-1} ( R \tilde{\beta} - q ) + o_p(1).
\end{equation}
\end{lemma}
\begin{proof}[\textbf{Proof of Lemma~\ref{Lemmafirstorderapprox}}]
We prove the lemma for the simple case $H_0: \beta_{p_mm} = 0$, that is, $R=(0,\ldots,0,1)$ and $q=0$. First, application of the delta method yields the following first-order approximation
\[ \sqrt{n} (\tilde{\beta} - \beta) = \frac{1}{\sqrt{n} \eta_0} \sum_{i=1}^n w(d_i(\beta,\tilde{\Sigma})) \Omega^{-1} x_i^\top \tilde{\Sigma}^{-1} e_i(\beta) + o_p(1), \]
where $d_i^2(\beta,\Sigma) = e_i^\top(\beta) \Sigma^{-1} e_i^\top(\beta)$, $e_i^\top(\beta) = \tilde{y}_i - x_i \beta$ and $\Omega = {\rm E}_K [X^\top \Sigma^{-1} X]$. If we replace $\tilde{\Sigma}$ with its true value $\Sigma$, we obtain an asymptotic equivalent expression. A similar expression is true for $\tilde{\beta}_r$. Decompose $\beta=(\beta^{(1)^{\scriptstyle t}},\beta^{(2)})^\top$ with $\beta^{(2)} = \beta_{p_mm}$ and similarly for $x_i$ and other variables. Then a first-order approximation for $\tilde{\beta}_r^{(1)}$ is given by
\[ \sqrt{n} (\tilde{\beta}_r^{(1)} - \beta^{(1)}) = \frac{1}{\sqrt{n} \eta_0} \sum_{i=1}^n w(d_i(\beta,\tilde{\Sigma}_r)) (\Omega^{(1)})^{-1} x_i^{(1)^{\scriptstyle t}} \tilde{\Sigma}_r^{-1} e_i(\beta) + o_p(1), \]
with $\Omega^{(1)} = {\rm E}_K [X^{(1)^{\scriptstyle t}} \Sigma^{-1} X^{(1)}]$, since under $H_0$ it holds that $e_i(\beta) = \tilde{y}_i - x_i^{(1)} \beta^{(1)}$.

In this simple case it is easy to show that the $p$th component of the right hand-side of equation~\eqref{firstorderapprox} is equal to $\tilde{\beta}^{(2)}$. Therefore, we only need to prove the lemma for the remaining components. Denote $P$ as the $(p-1) \times p$ elimination matrix, i.e., $P=(I_{p-1},(0,\ldots,0))$. Then, using the first-order approximations we obtain
\[ \sqrt{n} (\tilde{\beta}^{(1)} - \tilde{\beta}_r^{(1)}) = \frac{1}{\sqrt{n} \eta_0} \sum_{i=1}^n w(d_i(\beta,\Sigma)) \left[ P \Omega^{-1} x_i^\top - (\Omega^{(1)})^{-1} x_i^{(1)^{\scriptstyle t}} \right] \Sigma^{-1} e_i(\beta) + o_p(1). \]
Considering the general expression for the inverse of a block matrix, the terms between brackets reduce to
\[ P \Omega^{-1} R^\top (R \Omega^{-1} R^\top)^{-1} R \Omega^{-1} x_i^\top. \]
Consequently, we have that
\[ \sqrt{n} (\tilde{\beta}^{(1)} - \tilde{\beta}_r^{(1)}) = P \Omega^{-1} R^\top (R \Omega^{-1} R^\top)^{-1} R \frac{1}{\sqrt{n} \eta_0} \sum_{i=1}^n w(d_i(\beta,\Sigma)) \Omega^{-1} x_i^\top \Sigma^{-1} e_i(\beta) + o_p(1). \]
In the last line we recognize the first-order approximation of $\tilde{\beta}$. Hence,
\[ \sqrt{n} (\tilde{\beta}^{(1)} - \tilde{\beta}_r^{(1)}) = \sqrt{n} P \Omega^{-1} R^\top (R \Omega^{-1} R^\top)^{-1} \tilde{\beta}^{(2)} + o_p(1). \]
By combining these results the lemma is proven for the case $H_0: \beta_{p_mm} = 0$. To obtain the general result, we need to obtain a first-order approximation for $\tilde{\beta}_r$ starting from its (general) estimating equation and continue as above, but this derivation is quite lengthy and therefore is omitted.
\end{proof}

\begin{proof}[\textbf{Proof of Theorem~\ref{TheoremAsymptoticDistributionLambda}}]
Consider $\Lambda_{\text{S}}$ first. Expectations in the proof are with respect to $K$. Write $\tilde{y}_i = A_i^\top y$ and $x_i = A_i^\top X$ as above. Application of the delta method permits us to rewrite the test statistic as
\[ \Lambda_{\text{S}} = -n \ln \left( \frac{\abs{\tilde{\Sigma}}}{\abs{\tilde{\Sigma}_r}} \right) = - nm \ln \left( \frac{\tilde{\sigma}^2}{\tilde{\sigma}_r^2} \right) = nm \frac{\tilde{\sigma}_r^2 - \tilde{\sigma}^2}{\tilde{\sigma}_r^2} + o_p(1). \]
An alternative to~\eqref{MultivariateMscale} is to define $\tilde{s}(b,G)$ as the solution of
\begin{equation}
\label{MultivariateMscale2}
\frac{1}{n} \sum_{i=1}^n \rho_0 \left( \frac{\sqrt{e_i(b)^\top \phi(G^{-1}) e_i(b)}}{\tilde{s}(b,G)} \right) = \delta_0,
\end{equation}
with $\phi(A) = \abs{A}^{-1/m} A$ for a $m \times m$ matrix $A$ and where $e_i(b) = A_i^\top (y - X b) = \tilde{y}_i - x_i b$. Now, $G$ can be any positive definite matrix of size $m \times m$ (without imposing the restriction that $\abs{G} = 1$). Moreover, it holds that $\tilde{s}(\tilde{\beta},\tilde{\Gamma}) = \tilde{\sigma}$ and $\tilde{s}(\tilde{\beta}_r,\tilde{\Gamma}_r) = \tilde{\sigma}_r$. Hence,
\[ \Lambda_{\text{S}} = nm \frac{\tilde{s}^2(\tilde{\beta}_r,\tilde{\Gamma}_r) - \tilde{s}^2(\tilde{\beta},\tilde{\Gamma})}{\tilde{s}^2(\tilde{\beta}_r,\tilde{\Gamma}_r)} + o_p(1). \]
A Taylor expansion of $\tilde{s}^2(\tilde{\beta}_r,\tilde{\Gamma}_r)$ around $(\tilde{\beta},\tilde{\Gamma})$ yields
\begin{equation}
\label{eq1}
\tilde{s}^2(\tilde{\beta}_r,\tilde{\Gamma}_r) - \tilde{s}^2(\tilde{\beta},\tilde{\Gamma}) = \frac{1}{2} (\tilde{\beta}_r - \tilde{\beta})^\top \left( \frac{\partial^2 \tilde{s}^2(b,G)}{\partial b^\top \partial b} \right) \Big{|}_{(\tilde{\beta}^{*},\tilde{\Gamma}^{*})} (\tilde{\beta}_r - \tilde{\beta}) + o_p(1/n),
\end{equation}
where $\tilde{\beta}^{*}$ is an intermediate point between $\tilde{\beta}_r$ and $\tilde{\beta}$ and $\tilde{\Gamma}^{*}$ is an intermediate point between $\tilde{\Gamma}_r$ and $\tilde{\Gamma}$. Due to the definition of the S-estimator, the first-order derivatives vanish (see also the consistency condition of $\Lambda_{\text{S}}$). The second-order derivative of $G$ and the mixed derivative can be shown to be of order $o_p(1/n)$. Then we simplify the second-order derivative w.r.t.\ $b$. It can be shown that
\[ \left( \sum_{i=1}^n \frac{\psi_0(d_i(b,G)) d_i(b,G)}{\tilde{s}(b,G)} \right) \diff({\tilde{s}(b,G)},b) = - \left( \sum_{i=1}^n \frac{w_0(d_i(b,G))}{\tilde{s}^2(b,G)} x_i^\top \phi(G^{-1}) e_i(b) \right), \]
with $d_i^2(b,G) = e_i^\top(b) \phi(G^{-1}) e_i(b) / \tilde{s}^2(b,G)$. Taking derivatives w.r.t.\ $b$, we obtain
\[ \left( \frac{\partial^2 \tilde{s}(b,G)}{\partial b^\top \partial b} \right) \left( \sum_{i=1}^n \frac{\psi_0(d_i(b,G)) d_i(b,G)}{\tilde{s}(b,G)} \right) + \left( \diff({\tilde{s}(b,G)},b) \right) \left( \diff(,b) \sum_{i=1}^n \frac{\psi_0(d_i(b,G)) d_i(b,G)}{\tilde{s}(b,G)} \right)^\top, \]
for the left hand side and
\begin{multline*}
\sum_{i=1}^n \frac{w_0(d_i(b,G))}{\tilde{s}^2(b,G)} x_i^\top \phi(G^{-1}) x_i + \sum_{i=1}^n \frac{w_0'(d_i(b,G))}{d_i(b,G) \tilde{s}^4(b,G)} x_i^\top \phi(G^{-1}) e_i(b) e_i^\top(b) \phi(G^{-1}) x_i \\
+ \sum_{i=1}^n \frac{w_0(d_i(b,G)) + w_0'(d_i(b,G)) d_i(b,G)/2}{\tilde{s}^4(b,G)} x_i^\top \phi(G^{-1}) e_i(b) \left( \diff({\tilde{s}^2(b,G)},b) \right)^\top,
\end{multline*}
for the right hand side. Since $\tilde{\beta}$ and $\tilde{\beta}_r$ are consistent estimators (under $H_0$), also $\tilde{\beta}^{*}$ is consistent. Similarly for $\Gamma$. Therefore, since it is true that
\[ \diff({\tilde{s}(b,G)},b) \Big{|}_{(\tilde{\beta}^{*},\tilde{\Gamma}^{*})} \overset{a.s.}{\longrightarrow} 0, \]
we have
\[ \left( \frac{\partial^2 \tilde{s}(b,G)}{\partial b^\top \partial b} \right) \Big{|}_{(\tilde{\beta}^{*},\tilde{\Gamma}^{*})} \overset{a.s.}{\longrightarrow} \frac{\sigma \eta_0}{\gamma_0} {\rm E}_K [X^\top \Sigma^{-1} X]. \]
Then~\eqref{eq1} reduces to
\begin{equation}
\label{eq2}
\tilde{s}^2(\tilde{\beta}_r,\tilde{\Gamma}_r) - \tilde{s}^2(\tilde{\beta},\tilde{\Gamma}) = \frac{\sigma^2 \eta_0}{\gamma_0} (\tilde{\beta}_r - \tilde{\beta})^\top {\rm E}_K [X^\top \Sigma^{-1} X] (\tilde{\beta}_r - \tilde{\beta}) + o_p(1/n),
\end{equation}
and the test-statistic can be rewritten as
\[ \Lambda_{\text{S}} = \frac{nm \eta_0}{\gamma_0} (\tilde{\beta}_r - \tilde{\beta})^\top {\rm E}_K [X^\top \Sigma^{-1} X] (\tilde{\beta}_r - \tilde{\beta}) + o_p(1). \]
Using the result of lemma~\ref{Lemmafirstorderapprox} this expression for $\Lambda_{\text{S}}$ becomes
\[ \Lambda_{\text{S}} = \frac{nm \eta_0}{\gamma_0} ( R \tilde{\beta} - q )^\top \left( R \hspace{0.5mm} {\rm E}_K [X^\top \Sigma^{-1} X]^{-1} R^\top \right)^{-1} ( R \tilde{\beta} - q ) + o_p(1). \]
Using the results from Theorem~\ref{TheoremAsymptoticVariance}, we can rewrite this as
\[ \Lambda_{\text{S}} = \frac{n \alpha_0}{\eta_0 \gamma_0} ( R \tilde{\beta} - q )^\top \left( R \hspace{0.5mm} \text{ASV}(\tilde{\beta},H_{\beta,\Sigma}) R^\top \right)^{-1} ( R \tilde{\beta} - q ) + o_p(1). \]
Finally, the result follows by applying Slutzky's theorem.

Now, consider $\Lambda_{\text{MM}}$. The proof is similar as above, therefore, we only give a sketch of the proof. Write $\Lambda_{\text{MM}}$ as the difference of $\hat{\sigma}_r^2$ and $\hat{\sigma}^2$. Consider these estimates as a function of $\hat{\beta}$, $\hat{\Gamma}$ and $\tilde{\sigma}^2$ (and their restricted versions respectively). Then a Taylor expansion leads to similar expressions as in~\eqref{eq1} and~\eqref{eq2}. Since the result of Lemma~\ref{Lemmafirstorderapprox} can be generalized to MM-estimators, a similar derivation as above ends the proof.
\end{proof}


\begin{proof}[\textbf{Proof of Theorem~\ref{TheoremInfluenceFunctionsLambda}}]
Let $Z=(\tilde{X}^\top,Y^\top)^\top$ have model distribution $H=(K,F_\Sigma)$ with $K$ the distribution of $\tilde{X}$ and $F_\Sigma$ the elliptically symmetric distribution of the error terms. Let us denote $H_{\epsilon} = H_{\epsilon,\Delta_{z}}$ to simplify the notation. We only derive the second-order influence function of $\Lambda_{\text{S}}$ and $\Lambda_{\text{MM}}$ (see~\eqref{RobustLikelihoodRatioS} and~\eqref{RobustLikelihoodRatioM}) under $H_0: \beta_{p_mm}=0$.

First, consider $\Lambda_{\text{S}}$. We introduce its functional version as
\[ \Lambda_{\text{S}}(H) = - 2m \ln \left( \frac{\tilde{\sigma}(H)}{\tilde{\sigma}_r(H)} \right), \]
with $\tilde{\sigma}(H)$ the population version of $\tilde{\sigma}$. Under the null hypothesis we obtain
\[ \text{IF2}(z;\Lambda_{\text{S}},H) = \frac{2m}{\sigma} \left( \text{IF2}(z;\tilde{\sigma}_r,H) - \text{IF2}(z;\tilde{\sigma},H) \right), \]
by taking the second-order derivative with respect to $\epsilon$. Hence, the proof requires the second-order influence function of $\tilde{\sigma}$. For the scale functional, the following equation holds:
\[ \int \rho_0 \left( \frac{\sqrt{e(H)^\top \hat{\Gamma}^{-1}(H) e(H)}}{\tilde{\sigma}(H)} \right) dH(z) = \delta_0, \]
with $e(H) = y - x \hat{\beta}(H)$. To find $\text{IF2}(z;\tilde{\sigma},H)$, we consider this equation for $H=H_{\epsilon}$ and differentiate it twice. Since we need the difference of $\text{IF2}(z;\tilde{\sigma},H)$ and $\text{IF2}(z;\tilde{\sigma}_r,H)$, we only mind about the terms that are different, i.e., only the terms involving $\beta$ since these are different. Performing similar steps as in the proof of Theorem~\ref{TheoremInfluenceFunctions} we obtain
\[ \text{IF2}(z;\tilde{\sigma},H) = \frac{\eta_0 \sigma}{\gamma_0} \text{IF}(z;\tilde{\beta},H)^\top \hspace{0.5mm} \Omega \hspace{0.5mm} \text{IF}(z;\tilde{\beta},H) + R(\tilde{\Sigma}), \]
with $\Omega = {\rm E}_K [X^\top \Sigma^{-1} X]$ and where $R(\tilde{\Sigma})$ contains the remaining terms not involving $\beta$. Remark that $R(\tilde{\Sigma})$ has an explicit expression, but to save space we do not show it. Likewise, such an expression can be obtained for $\tilde{\sigma}_r$. Consequently, we have
\[ \text{IF2}(z;\Lambda_{\text{S}},H) = - \frac{2m \eta_0}{\gamma_0} \text{IF}(z;\tilde{\beta},H)^\top \hspace{0.5mm} \Omega \hspace{0.5mm} \text{IF}(z;\tilde{\beta},H) + \frac{2m \eta_0}{\gamma_0} \text{IF}(z;\tilde{\beta}_r^{(1)},H)^\top \hspace{0.5mm} \Omega^{(1)} \hspace{0.5mm} \text{IF}(z;\tilde{\beta}_r^{(1)},H), \]
with $\Omega^{(1)} = {\rm E}_K [X^{(1)^{\scriptstyle t}} \Sigma^{-1} X^{(1)}]$, where $X^{(1)}$ is defined as in Lemma~\ref{Lemmafirstorderapprox}. Plugging in the results of Theorem~\ref{TheoremInfluenceFunctions} we get
\[ \text{IF2}(z;\Lambda_{\text{S}},H) = - \frac{2m}{\eta_0 \gamma_0} w_0^2(\norm{e(H)}_{\Sigma}) e(H)^\top \Sigma^{-1} \left[ x \Omega^{-1} x^\top - x^{(1)} (\Omega^{1})^{-1} x^{(1)^{\scriptstyle t}} \right] \Sigma^{-1} e(H), \]
where $x^{(1)}$ is defined similarly as $X^{(1)}$. A comparable reasoning as in Lemma~\ref{Lemmafirstorderapprox} shows that the previous line reduces to
\[ - \frac{2m}{\eta_0 \gamma_0} w_0^2(\norm{e(H)}_{\Sigma}) e(H)^\top \Sigma^{-1} x \Omega^{-1} R^\top \left( R \Omega^{-1} R^\top \right)^{-1} R \Omega^{-1} x^\top \Sigma^{-1} e(H). \]
Now we recognize the influence function of $\tilde{\beta}$ and find
\[ \text{IF2}(z;\Lambda_{\text{S}},H) = - \frac{2m \eta_0}{\gamma_0} \text{IF}(z;\tilde{\beta},H)^\top R^\top \left( R \Omega^{-1} R^\top \right)^{-1} R \hspace{0.5mm} \text{IF}(z;\tilde{\beta},H), \]
as was to be proven.

Then, consider $\Lambda_{\text{MM}}$ with functional version
\[ \Lambda_{\text{MM}}(H) = - 2m \ln \left( \frac{\hat{\sigma}(H)}{\hat{\sigma}_r(H)} \right), \]
with $\hat{\sigma}(H)$ the population version of $\hat{\sigma}$. Under the null hypothesis the second-order influence function of $\Lambda_{\text{MM}}$ is again the difference of the second-order influence functions of $\hat{\sigma}$. An identical derivation verifies the result.
\end{proof}


\begin{proof}[\textbf{Proof of Theorem~\ref{TheoremAsymptoticDistributionLM}}]
We proof the result for $\text{LM}_{\text{MM}}$. Consider a SUR model with two blocks ($m=2$) for ease of notation. The results can be generalized to $m>2$.

Since the estimating equations of $\hat{\Sigma}_r$ are the diagonal parts of equations~\eqref{EstimatingEquationMM}, $\text{LM}_{\text{MM}}$ can be rewritten as
\[ \text{LM}_{\text{MM}} = \frac{n m^2 \left( \sum_{i=1}^n w_1(d_i) e_{i1} e_{i2} \right)^2}{\hat{\sigma}_{r,11} \hat{\sigma}_{r,22} \left( \sum_{i=1}^n \psi_1(d_i) d_i \right)^2}, \]
with $\hat{\sigma}_{r,jj}$ the $j$th diagonal element of $\hat{\Sigma}_r$. According to the null hypothesis we have
\[ {\rm E}_F [w_1(d) e_1 e_2] = 0, \]
with $d^2=e^\top \Sigma^{-1} e$ and $e=(e_1,e_2)^\top \sim F = F_{\Sigma}$. Moreover, due to a result of~\citet{Lopuhaa1999}, for the variance we obtain
\[ {\rm Var}_F [w_1(d) e_1 e_2] = \frac{\sigma_{11} \sigma_{22}}{m(m+2)} {\rm E}_F [\psi_1^2(\norm{e}_{\Sigma}) \norm{e}_{\Sigma}^2]. \]
Since $w_1(d_i) e_{i1} e_{i2}$, $i=1,\ldots,n$ are independent identically distributed random variables, the central limit theorem gives
\[ \sqrt{n} \left( \frac{1}{n} \sum_{i=1}^n w_1(d_i) e_{i1} e_{i2} \right) \quad \overset{d}{\longrightarrow} \quad N(0,{\rm Var}_F [w_1(d) e_1 e_2]), \]
or equivalently
\[ \frac{nm(m+2)}{\sigma_{11} \sigma_{22} {\rm E}_F [\psi_1^2(\norm{e}_{\Sigma}) \norm{e}_{\Sigma}^2]} \left( \frac{1}{n} \sum_{i=1}^n w_1(d_i) e_{i1} e_{i2} \right)^2 \quad \overset{d}{\longrightarrow} \quad \chi_1^2. \]
Since for $j=1,2$, $\hat{\sigma}_{r,jj}$ is a consistent estimator under $H_0$ and
\[ \frac{1}{n} \sum_{i=1}^n \psi_1(d_i) d_i \quad \overset{a.s.}{\longrightarrow} \quad \gamma_1, \]
the result now follows.
\end{proof}


\begin{proof}[\textbf{Proof of Theorem~\ref{TheoremInfluenceFunctionsLM}}]
Let $Z=(\tilde{X}^\top,Y^\top)^\top$ have model distribution $H=(K,F_\Sigma)$ with $K$ the distribution of $\tilde{X}$ and $F_\Sigma$ the elliptically symmetric distribution of $Y$. Let us denote $H_{\epsilon} = H_{\epsilon,\Delta_{z}}$ to simplify the notation. We derive the second-order influence function of $\text{LM}_{\text{MM}}$. Again we assume $m=2$ for simplicity.

We introduce the functional version as
\[ \text{LM}_{\text{MM}}(H) = \frac{\left( \int w_1(d(H)) e_1(H) e_2(H) \hspace{0.5mm} dH(z) \right)^2}{\left( \int w_1(d(H)) e_1^2(H) \hspace{0.5mm} dH(z) \right) \left( \int w_1(d(H)) e_2^2(H) \hspace{0.5mm} dH(z) \right)}, \]
with $d^2(H)=e^\top(H) \hat{\Sigma}^{-1}(H) e(H)$ and $e=(e_1,e_2)^\top$. Since the first-order influence function is exactly zero and
\[ \int w_1(d(H)) y_1 y_2 \hspace{0.5mm} dH(z) = 0, \]
under the null hypothesis, we obtain
\[ \text{IF2}(z;\text{LM}_{\text{MM}},H) = \frac{2m^2}{\gamma_1^2 \sigma_{11} \sigma_{22}} \left( \diff(,\epsilon) \left( \int w_1(d(H_{\epsilon})) e_1(H_{\epsilon}) e_2(H_{\epsilon}) \hspace{0.5mm} dH_{\epsilon}(z) \right) \Bigg{|}_{\epsilon = 0} \right)^2. \]
Following the same steps as in the proof of Theorem~\ref{TheoremInfluenceFunctions} the above derivative becomes
\[ \int \frac{w_1'(\norm{y}_{\Sigma})}{2\norm{y}_{\Sigma}} (-2y \Sigma^{-1} x \text{IF}(z;\hat{\beta},H) + y^\top \text{IF}(z;\hat{\Sigma}^{-1},H) y) y_1 y_2 \hspace{0.5mm} dH(z) - \int w_1(d(H)) y_1 y_2 \hspace{0.5mm} dH(z) + w_1(\norm{y}_{\Sigma}) y_1 y_2. \]
Due to symmetry and under $H_0$ the integrals vanish. Combining the results, proves the theorem.
\end{proof}

\end{document}